\documentclass[10pt]{article}

\usepackage{cite}
\usepackage{amsmath,amssymb,amsfonts}
\usepackage{algorithmic}

\usepackage{textcomp}
\usepackage{lipsum}
\usepackage{float}
\usepackage{mathtools}
\usepackage{mathrsfs}
\usepackage{amssymb}
\usepackage{booktabs}
\usepackage{multirow}
\usepackage{subfig}
\usepackage{graphicx}
\usepackage{epstopdf}
\usepackage{breqn}
\usepackage{xcolor}
\usepackage{enumerate}
\usepackage{breqn}
\usepackage[normalem]{ulem}
\usepackage{arydshln}

\usepackage{enumitem}

\usepackage{amsmath}
\usepackage{cite}
\usepackage{refcount}
\usepackage{flushend}
\usepackage{amsthm}
\definecolor{DarkBlue}{rgb}{0.1,0.22,0.45}
\definecolor{DarkGreen}{rgb}{0.1,0.22,0.45}
\definecolor{DarkRed}{rgb}{0.1,0.22,0.45}

\usepackage[hyperfootnotes=false,colorlinks = true, linkcolor = black,citecolor = black]{hyperref}
\usepackage[capitalize]{cleveref}
\setlist{nolistsep}

\newtheorem{thm}{Theorem}[section]
\newtheorem{deff}{Definition}[section]
\newtheorem{rmk}{Remark}[section]
\newtheorem{lem}{Lemma}[section]
\newtheorem{exmp}{Example}[section]
\newtheorem{cor}{Corollary}[section]

\newtheorem{property}{Property}[section]
\newtheorem{assumption}{Assumption}[section]

\newenvironment{newproofof}{\itshape{Proof of}}{\hfill$\square$}

\usepackage[margin=1.7in]{geometry}

\usepackage{fancyhdr}

\newcommand\shorttitle{ON THE LOCAL INPUT-OUTPUT STABILITY OF EVENT-TRIGGERED CONTROL SYSTEMS}
\newcommand\authors{MOHSEN GHODRAT AND HORACIO J MARQUEZ}

\begin{document}

\title{\normalsize\textbf{ON THE LOCAL INPUT-OUTPUT STABILITY OF EVENT-TRIGGERED CONTROL SYSTEMS}}

\author{\small MOHSEN GHODRAT AND HORACIO J MARQUEZ 
\thanks{The authors are with the Department of Electrical and Computer Engineering, University of Alberta, Edmonton, AB T6G 2V4, Canada (e-mail:
ghodrat@ualberta.ca; marquez@ece.ualberta.ca)}}
\date{}

\maketitle

\begin{abstract} \small\baselineskip=9pt This paper studies performance preserving event design in nonlinear event-based control systems based on a local $\mathcal{L}_2$-type performance criterion. Considering a finite gain local $\mathcal{L}_2$-stable disturbance driven continuous-time system, we propose a triggering mechanism so that the resulting sampled-data system  preserves similar disturbance attenuation local $\mathcal{L}_2$-gain property. The results are applicable to nonlinear systems with exogenous disturbances bounded by some Lipschitz-continuous function of state. It is shown that an exponentially decaying function of time, combined with the proposed triggering condition, extends the inter-event periods. Compared to the existing works, this paper analytically estimates the increase in intersampling periods at least for an arbitrary period of time.  We also propose a so-called \emph{discrete triggering condition} to quantitatively find the improvement in inter-event times at least for an arbitrary number of triggering iterations. Illustrative examples support the analytically derived results. 
\end{abstract}


\section{Introduction}

{E}vent-based control systems have been an active area of research over the last decade. The primary characteristic of event-based controllers is that they can provide performance very similar to classical control approaches while reducing the transmission of information between plant and controller. This feature is important in a growing number of applications in which limiting transmission rates is a concern. Examples include
battery-operated systems with wireless transmission between plant and controller, which often have limited energy and/or memory supplies, or network control systems with shared wired or wireless communication channels, \cite{hespanha_survey}. In a (classical) time-triggering fashion, data transmission between system elements (such as actuator, sensor and plant) occurs periodically, regardless of whether or not changes in the measured output and/or commands require computation of a new control output. In an event-based scenario, the system decides when to update the control output based on a so called \emph{real time triggering condition} on the measured signals. 

Event-based systems have been used without theoretical supports for many years. The resurgence of interest in the subject began with the work reported in reference \cite{astrom_stochastic} that considers a first order stochastic system and shows that event-based sampling offers better performance than classical time-triggered control in terms of closed loop variance and sampling rate.
Following publication of this work, event-triggered systems became a very active area of research and many important contributions have been reported addressing \emph{stability} (\cite{PID,tabuada, Heemels_periodic,passive,Postuyan_unifying,Postuyan_GFW}), and \emph{performance} (\cite{tracking,minimum_attention,L2_Wang_conf,L_2_selftriggered,L_2_selftriggered_disturbance,Dolk_Lp,Dolk_LP_ieee}), to mention a few (see also the references therein). 

Reference \cite{PID}, one of the first references on stabilization of event-based systems, proposes an event-based mechanism for PID control. 
Reference \cite{tabuada}, presents a clever and rather general solution to the stability problem of event-triggered systems. In this reference the author assumes the existence of a pre-designed continuous-time control law that results in input-to-state stability of a nonlinear plant, and shows
that restricting the measurement error ({\it i.e.}, the difference between the system state and the last sampled value) to stay within a function of the state threshold, guarantees closed loop global asymptotic stability.

Reference \cite{tabuada}  has inspired much work and several event-based strategies have been proposed that extend this work (see  \cite{tabuada_survey} and the references therein).
Reference \cite{tabuada} is restricted to state-feedback and therefore relies on full state measurement. This restriction is relaxed in \cite{Heemels_periodic,passive}. Reference \cite{Heemels_periodic}, considers \emph{periodic} event-based control of linear systems, in which the triggering condition is monitored at regular intervals instead of continuously, and can be viewed as a sampled data version of event-triggered systems. Reference \cite{passive} considers output feedback stabilization using the framework of passivity theory. References \cite{Postuyan_unifying,Postuyan_GFW} offer a unifying framework for the stability problem of nonlinear event-based
in the context of hybrid systems.

All of the above mentioned works focus on stabilization. The effects of an event-based mechanism on control performance was first addressed in \cite{minimum_attention}, which shows a trade-off between system performance and the complexity of the control law. 
A decentralized event-triggered mechanism is proposed in \cite{L_infty_gain} for distributed linear systems. This reference considers an 
impulsive system approach to system stability and proposes an event-based mechanism that satisfies an  $\mathcal{L}_\infty$ bound.
References  \cite{lemmonlinearL_2fullstate,Lemmonbook,L_2_selftriggered, L_2_selftriggered_disturbance, L2_multiagent,L2_Wang_conf,passive_input_output,passive_delay} focus on the
$\mathcal{L}_2$-gain.
The $\mathcal{L}_2$-gain stability analysis of event-based systems was first investigated in \cite{lemmonlinearL_2fullstate} where a full-information $\mathcal{H}_\infty$ controller is proposed for LTI systems. \cite{Lemmonbook,L_2_selftriggered} continued the work of \cite{lemmonlinearL_2fullstate} in more details. \cite{Lemmonbook} proposes an $\mathcal{L}_2$-gain performance-preserving triggering condition for a class of nonlinear affine systems. \cite{L_2_selftriggered} considers LTI systems and derives an explicit lower bound on the sampling periods. In this reference, the disturbance is assumed to be norm bounded by a linear function of the state norm. This condition is then relaxed in \cite{L_2_selftriggered_disturbance}. 
Reference \cite{L2_multiagent} considers the $\mathcal{L}_2$-gain of distributed multi agent systems under event-triggered agreement protocols. 
Reference \cite{L2_Wang_conf} proposes an 
event-triggered mechanism for distributed network linear systems and guarantees finite gain $\mathcal{L}_2$-stability  in the presence of packet data dropouts. 
Reference \cite{passive_input_output} considers passive systems and proposes a triggering condition that guarantees finite gain $\mathcal{L}_2$-stability when the external disturbance is bounded and shows that their approach preserves stability under constant network induced delays
or delays with bounded jitters. Reference \cite{passive_delay} extends the results of \cite{passive_input_output} to systems with constant network induced delays
or time-varying delays with bounded jitters. \cite{Dolk_Lp} proposes a dynamic triggering condition for the centralized state feedback event-based control of nonlinear network control systems with guaranteed $\mathcal{L}_p$-stability. Reference \cite{Dolk_LP_ieee} extends the work of \cite{Dolk_Lp} to the output feedback and decentralized case. 

This note tackles the finite gain $\mathcal{L}_2$-stability problem of event-triggered systems. Our primary interest is to propose a novel triggering scheme that, compared to the above mentioned works, solves the $\mathcal{L}_2$-gain preserving event design problem for a wider class of systems. Indeed, we consider a rather general class of nonlinear system model with the sole assumption of satisfying a rather mild local Lipschitz continuity condition. Taking exogenous disturbances together with measurement errors as inputs, our proposed triggering condition is obtained based on the assumption that the system is input-to-state stable (ISS). 
We assume that the disturbance term is originated from structural uncertainties in the system model and is norm bounded by some locally Lipschitz-continuous function of state. This assumption is rather mild and more general than previous references. For example, in the framework of self-triggered control, \cite{L_2_selftriggered} considers a similar $\mathcal{L}_2$ problem to the one studies here, but assumes that the norm of disturbance is bounded by a linear function of the state norm. Unfortunately, dealing with uncertainties in the event-triggered context is non-trivial. Indeed the triggering mechanism is designed to update the actuators whenever measurement errors are above a pre-established threshold. In the absence of disturbances, the error originates during the intersample as the difference between the present value of the state and its last sampled value. 
In the presence of exogenous disturbance, however, the error is also driven by the disturbance term making it difficult to design an effective triggering condition.

The ISS assumption implies working with bounded inputs and therefore 
suggests the need to consider \emph{small signals} in some sense. To formalize this concept, we present our results
using
an extension of the classical input-output theory of systems with modified input spaces, referred to as {\it local} (or {\it small signal}) input-output stability introduced in  \cite{marquez_local}.

Therefore we can state our contributions as follows: 
First, given a nonlinear plant and a previously designed full-information controller that satisfies a ({\it local}) $\mathcal{L}_2$ performance bound, we present conditions under which the same norm is guaranteed when the controller is implemented in an event-triggered fashion. The importance of this problem lies on the fact that many control systems can be designed to satisfy a tight $\mathcal{L}_2$-gain condition for small-size input signals, although the same gain may not hold when signal amplitude becomes large. We also show that the resulting sequence of sampling times is a uniformly isolated sequence, and hence inter-event periods are strictly nonzero. Finally, we show that, in the absence of disturbances, the system is globally asymptotically stable.

Our second contribution consists of a modified triggering condition to limit the average frequency of triggerings without violating the $\mathcal{L}_2$ property of the resulting system. Indeed, reducing the number of triggering instants over a time interval (average sampling frequency) may be important. For example, high data load on a communication channel over a finite time interval may result in undesirable challenges such as data packet drop out and/or transmission delay. Compared to the existing event-triggered and time-regularized strategies, our method is shown to be efficient in terms of transmission rate, even when the system trajectories are close to the origin, see, {\it{e.g.,}} Examples \ref{exmple2}, \ref{example3}. In addition, our strategy improves existing results, see, {\it e.g.}, \cite{event_time,Girard,me_iet,Dolk_Lp,Dolk_LP_ieee,Postuyan_GFW}, in that the increase in inter sampling-times can be designed a-priori, at least for a desired period of time (or a desired number of triggering iterations). By contrast, in the approaches in \cite{event_time,Girard,me_iet,Dolk_Lp,Dolk_LP_ieee,Postuyan_GFW} the intersampling increase is not estimated quantitatively. We also show that there is a trade-off between intersampling improvement and stability of zero-input system, in the sense that enlarging intersampling periods results in practical sense stability rather than the classical notion of stability.

Our work can be distinguished from those of \cite{lemmonlinearL_2fullstate,Lemmonbook,L_2_selftriggered,L_2_selftriggered_disturbance,L2_multiagent,passive_delay,passive_input_output,Dolk_Lp,Dolk_LP_ieee} as follows: First, we provide a general treatment of input-output performance preservation for event-based feedback systems. Restricting inputs to the space of small size signals, in the sense defined later, we investigate finite gain local $\mathcal{L}_2$-stability of a system with a controller implemented using an event-based mechanism when the original continuous-time system is also finite gain locally $\mathcal{L}_2$-stable. As we will show later, the need of the small signal approach in our study arises from the ISS assumption on the event-based system. A similar but non-local input-output stability problem is considered in the aforementioned references but for a more restrictive class of systems. References \cite{lemmonlinearL_2fullstate,L_2_selftriggered,L_2_selftriggered_disturbance} address the $\mathcal{L}_2$-stability problem of linear systems. \cite{Lemmonbook} extends these results to a class of nonlinear \emph{affine} systems. 
\cite{L_2_selftriggered,L_2_selftriggered_disturbance,L2_multiagent,passive_delay,passive_input_output} consider $\mathcal{L}_2$-gain performance without restriction on the input space ({\it i.e.} the non-local problem).
We frame our work in the context of the  \emph{dissipativity} introduced in \cite{Willems}, but restricting the class of input functions in a way that fits our needs for a local stability theory.

Other references employing dissipative include reference \cite{passive} that studies (non-local) stability of passive systems. References \cite{passive_input_output}, \cite{passive_delay} extend and generalize the results of \cite{passive} to systems with disturbances to guarantee finite gain $\mathcal{L}_2$-stability of the passive system. In these works, the authors assume that the (dynamic) controller communicates continuously to the actuators.
In \cite{L_infty_gain}, it is shown how this assumption can be relaxed by introducing a second triggering mechanism after the controller output. Here we avoid continuous data transfer between controller and plant by choosing a static controller. 
Also compared to \cite{passive_input_output}, \cite{passive_delay}, we address the problem in a more general nonlinear setup. Indeed, except the more general assumption of output feedback, these references focus on passive systems which will be covered in our work as a special case. Furthermore, the desired output in these references is assumed to be the same as the measured output that is also relaxed in our formulation.  

References \cite{Dolk_Lp,Dolk_LP_ieee} propose a time-regularized triggering method for the $\mathcal{L}_p$-stability of nonlinear systems suggesting that the triggering condition is checked only after some specific time has passed since the most recent triggering instant. In this paper we consider a fully event-triggered mechanism which enjoys several advantages compared to the time-regularized approach, {\it{e.g.,}} when the trajetories of event-based control system converge to the origin, the time-regularized triggering often reduces to traditional periodic sampling (see {\it e.g.,} \cite{Dolk_LP_ieee,Dolk_Lp,event-separation} and the examples therein). However, this undesirable behavior is not present under our proposed triggering strategy.

The remainder of the paper is organized as follows. We begin with introducing our notations and preliminary definitions needed throughout the paper. In section \ref{section problem statement} we formulate the main problem to be solved. In sections \ref{section main result} we provide the main results on the local $\mathcal{L}_2$-gain stability of nonlinear systems with state-dependent disturbances. In section \ref{section improve inter-execution time} we introduce a new triggering condition based on the structure of our problem statement to improve the inter-execution times. We also find the amount of increase in the inter-event times' lower bound compared to the old scenario. Illustrative examples are given in section \ref{section examples} to support the results of sections \ref{section main result}. The proofs are mostly given in the Appendix except the short ones that are provided in the body text.

\section{Preliminaries and Problem Statement}\label{section problem statement}

\subsection{Notation and Definitions}
Throughout the paper  $\mathbb{R}$ and $\mathbb{Z}$ represent the field of real numbers and the set of integers, respectively. 
$\mathbb{R}_{\geq0}$, $\mathbb{Z}_{\geq0}$, $\mathbb{R}_{>0}$ 
and $\mathbb{Z}_{>0}$ are the sets of nonnegative and positive elements of $\mathbb{R}$ and $\mathbb{Z}$.
$\mathbb{R}^n$ is the set of n-dimensional vectors with elements in $\mathbb{R}$ and $|x|$ is the Euclidean norm of column vector $x\in\mathbb{R}^n$. ${|x|}_{\infty}$ represents the infinity norm of vector $x=(x_1,x_2,\ldots,x_n)^T\in\mathbb{R}^n$, ${|x|}_{\infty}=\smash{\displaystyle{\max}_{1\leqslant i\leqslant n}\{|x_i|\}}$ and 
$\mathbb{R}_{\geq0}^2$ the causal triangular sector of $\mathbb{R}^2$ defined as $\mathbb{R}_{\geq0}^2=\{(x_2,x_1)\in\mathbb{R}^2|x_2\geqslant x_1\}$. 
By $[x^T~y^T]^T$ we denote the \emph{stack} column vector in $\mathbb{R}^{n+m}$, where $x\in \mathbb{R}^n$ and $y\in \mathbb{R}^m$. 
${||w||}_{2}$ and ${||w||}_{\infty}$  denote the $2$-norm and supremum-norm of function $w$, respectively, defined as ${||w||}_{2} =  (\int_{0}^{\infty}{|w(t)|^{2} dt})^{\frac{1}{2}}$ and ${||w||}_{\infty} =  \operatorname{ess}\sup_{t\in \mathbb{R}_{\geq 0}}|w(t)|$. $\mathcal{L}_2(\mathbb{R}_{\geq0})$ is the space of measurable functions $w$ with bounded $2$-norm. Also $\mathcal{L}_{\infty}(\mathbb{R}_{\geq0})$ is the space of measurable functions $w$ with bounded supremum norm. 
We say that a function is of class ${\bf{C}}^0$ (respectively ${\bf{C}}^1$) if it is continuous (respectively continuously differentiable). 
A function $f:\mathbb{R}^{n} \mapsto\mathbb{R}^p$ is said to be locally Lipschitz-continuous in an open set $B$, if for each $z\in B$ there exist $L_f\in\mathbb{R}_{>0}$ and $r\in\mathbb{R}_{>0}$ such that $|f({{x}})-f(\tilde{{x}})|\leqslant L_f|{{x}}-{\tilde{{x}}}|$ for all ${{x}},{\tilde{{x}}}\in \{y\in B|~|y-z|<r\}$. 
We also say that $f$ is Lipschitz-continuous in a set $D$ if there exists $L_f\in\mathbb{R}_{>0}$ (called the Lipschitz constant of $f$ on $D$) such that $|f({{x}})-f(\tilde{{x}})|\leqslant L_f|{{x}}-{\tilde{{x}}}|$ for all ${{x}},{\tilde{{x}}}\in D$. A function $\alpha:[0,a)\mapsto \mathbb{R}_{\geq 0}$
belongs to class $\mathcal{K}$ if it is strictly increasing and $\alpha(0)=0$. A class $\mathcal{K}$ function $\alpha$ belongs to class $\mathcal{K}_{\infty}$ if $a=\infty$ and $\alpha(r)\rightarrow \infty$ as $r\rightarrow \infty$.  

Throughout the paper we consider a nonlinear system $\mathscr{G}$ defined as follows: 
\begin{equation}\label{eq sys}
\mathscr{G}:
\begin{cases}
\dot{x}=f(x,u,w)\\
z=h(x,w)
\end{cases} 
\end{equation}
where $x\in\mathbb{R}^n$ represents the state, $u\in\mathscr{U}\subseteq\mathbb{R}^m$ the control input,
$w\in \mathscr{W}\subseteq\mathbb{R}^q$ the exogenous disturbance, and $z\in\mathbb{R}^p$ the measured output. 
We assume that $f$ and $h$ are class ${\bf{C}}^0$ and $f(0,0,0)=0$, $h(0,0)=0$ so that $x=0$ is an equilibrium point of zero-input system. Moreover, we will assume the state $x$ evolves on an open subset of $\mathbb{R}^n$ containing the origin.
We also assume that $\mathscr{G}$ is driven from initial conditions $x_0=x(t_0)$ and the inputs $u$ and $w$ are applied at time $t = t_0^+$. The state transition function for the system $\mathscr{G}$ is the function $\Phi: \mathbb{R}^2_{\geq0}\times\mathbb{R}^n\times\mathscr{U}\times \mathscr{W}\mapsto\mathbb{R}^n$ satisfying $x_0=\Phi(t_0,t_0,x_0,u,w)$ and $x(t)=\Phi(t_0,t,x_0,u,w)$ for all $(t,t_0)\in \mathbb{R}^2_{\geq0}$, $x_0\in\mathbb{R}^n$ and $w\in \mathscr{W}$.
\\
Input-output stability is a key tool in the rest of this note. The classical definitions of the input-output stability can be found in many references, see, {\it e.g.,} \cite{vidyasagar}. However, the results are not applicable to the systems with norm bounded input space. Instead, we build our theory using the local version of input-output stability introduced in \cite{marquez_local} and summarized as follows:
\\
In the next definitions we exploit the concept of \emph{relations} as a traditional tool to state the local stability criteria. Equivalently, one can define the input-output stability as a property of the operators. We recall that given two nonempty sets $A_1$ and $A_2$, a relation $\mathscr{R}$ on $A_1\times A_2$ is any subset of the Cartesian product $A_1\times A_2$.
	\begin{deff}
		Let $A_1\times A_2$ be the Cartesian product of two sets $A_1$ and $A_2$. We denote by $P_i:A_1\times A_2\rightarrow A_i$, $i=1,2$ the \emph{evaluation map} at $i$ defined as $P_i(x_1,x_2)=x_i$, $i=1,2$.  	
	\end{deff}
	\begin{deff}\label{wq}
		We define the set $\mathscr{W}_Q\subset\mathcal{L}_2(\mathbb{R}_{\geq0})$ as follows:
		\begin{equation}\label{def_W_Q}
			\mathscr{W}_Q=\{w\in\mathcal{L}_2(\mathbb{R}_{\geq0})|{\|w\|}_{\infty}<Q\},
		\end{equation}
		where $Q\in\mathbb{R}_{>0}$. We note that $\mathscr{W}_Q$, which is a subset of $\mathcal{L}_2(\mathbb{R}_{\geq0})\cap\mathcal{L}_\infty(\mathbb{R}_{\geq0})$, is not a linear space in general since there exists elements $x,y\in\mathscr{W}_Q$ such that $x+y\notin\mathscr{W}_Q$.
	\end{deff}
	We remark that the triplet $(\mathcal{L}_2(\mathbb{R}_{\geq 0}),{\|\cdot\|}_2,{\|\cdot\|}_{\infty})$ consisting of linear space $\mathcal{L}_2(\mathbb{R}_{\geq 0})$ and the norms ${\|\cdot\|}_{2}$ and ${\|\cdot\|}_{\infty}$ is a {\it binormed linear space}, where ${\|\cdot\|}_{2}$ and ${\|\cdot\|}_{\infty}$ are the {\it primary} and {\it secondary} norms of the space $\mathcal{L}_2(\mathbb{R}_{\geq 0})$. $\mathscr{W}_Q$ is then the subset of $\mathcal{L}_2(\mathbb{R}_{\geq 0})$ consisting of functions with secondary norm less than $Q\in\mathbb{R}_{>0}$.
	\begin{deff}
		A relation $\mathscr{R}$ on $\mathcal{L}_2(\mathbb{R}_{\geq0})\times \mathcal{L}_2(\mathbb{R}_{\geq0})$ is said to be $\mathscr{W}_Q$-stable if the evaluation map at $2$ is a bounded subset of $\mathcal{L}_2(\mathbb{R}_{\geq0})$ whenever the evaluation map at $1$ belongs to the set $\mathscr{W}_Q$.
	\end{deff}
	\begin{deff}
		The system $\mathscr{G}$ defined in (\ref{eq sys}) is said to be {\it locally $\mathcal{L}_2$-stable} if for any $w\in\mathscr{W}_Q$, the relation $\mathscr{R}\doteq \{(w,z)\in\mathcal{L}_2(\mathbb{R}_{\geq0})\times \mathcal{L}_2(\mathbb{R}_{\geq0})\}$ 
		is $\mathscr{W}_Q$-stable.
	\end{deff}
In the next definition, we provide a local version of finite gain $\mathcal{L}_2$-stability \footnote{See \cite{vidyasagar} for the classical finite gain $\mathcal{L}_2$-stability definition.}, a deviation from the classical definition by restricting the spaces of admissible inputs and initial conditions to the sets $\mathscr{W}_Q$ (defined in Definition \ref{wq}) and
\begin{eqnarray}\label{initial_cond_set}
\mathscr{X}_0\doteq  \{r\in\mathbb{R}^n|~|r|\leqslant\varepsilon\in\mathbb{R}_{>0}\},
\end{eqnarray}
respectively. 
\begin{deff}\label{def gain}The system $\mathscr{G}$ described in (\ref{eq sys}) is said to be finite gain locally $\mathcal{L}_2$-stable and has the local $\mathcal{L}_2$-gain less than or equal to $\Gamma$, if it is locally $\mathcal{L}_2$-stable and there exist finite constants $\eta\in \mathbb{R}_{\geq0}$, $\Gamma\in\mathbb{R}_{>0}$ and positive semi-definite ${\bf{C}}^0$ function $\mu$ such that for any $(T,t_0)\in\mathbb{R}_{\geq0}^2$, any $w\in\mathscr{W}_{Q}$ and any $x_0\in\mathscr{X}_0\subset\mathbb{R}^n$
\begin{equation}\label{practical eq}
\int_{t_0}^{T}{|z(s)|^2}ds\leqslant\Gamma^2\int_{t_0}^{T}{|w(s)|^2}ds+\mu(x_0)+\eta,
\end{equation}
where $z(s)=h(x(s),w(s))$, $x(s)=\Phi(t_0,s,x_0,u,w)$. We shall denote the local $\mathcal{L}_2$-gain of system $\mathscr{G}$ by ${\|\mathscr{G}\|}_{\mathcal{L}_2}$. We also say that $\mathscr{G}$ is finite gain locally $\mathcal{L}_2$-stable with zero bias if $\eta=0$ in (\ref{practical eq}).	
\end{deff}
The following theorem provides a sufficient condition to estimate an upper bound on the local disturbance attenuation $\mathcal{L}_2$-gain of system $\mathscr{G}$ in the context of dissipative systems theory introduced by \cite{Willems}.
\begin{thm}\label{thm gain}
	The nonlinear system $\mathscr{G}$ is finite gain locally $\mathcal{L}_2$-stable with zero bias and has ${\|\mathscr{G}\|}_{\mathcal{L}_2}\leqslant {\Gamma}$, provided there exist a positive definite ${\bf{C}}^1$ function $V$ and a control input $u\in\mathscr{U}$ such that for all $w \in \mathscr{W}_{Q}$
	\begin{equation}\label{eq H}
		H_{\Gamma}(V,u)\doteq \nabla V(x)\cdot f(x,u,w)-{\Gamma}^2 |w|^2+|h(x,w)|^2 \leqslant 0.
	\end{equation}
\end{thm}

\begin{proof}
	The result is readily obtained by integration of (\ref{eq H}), positive definiteness of $V(x)$ and Definition \ref{def gain}. 
\end{proof}
\begin{rmk}
	If system $\mathscr{G}$ is reachable from $x_0$, condition (\ref{eq H}) is necessary and sufficient for finite gain local $\mathcal{L}_2$-stability of $\mathscr{G}$ with zero bias and ${\|\mathscr{G}\|}_{\mathcal{L}_2}\leqslant {\Gamma}$.  
\end{rmk}

\begin{proof}
	The result follows directly from (\cite{james}, Theorem 2.1). 
\end{proof}

Next we investigate the input-to-state stability of system $\mathscr{G}$. We shall assume that the measurement of state is affected by an error $e$. As a result, designing the state feedback controller $u=k(x)$, where $k$ is of class ${\bf{C}}^0$ and satisfy $k(0)=0$, the implemented control law will be $k(x+e)$. The corresponding closed loop system with perturbed measurement is therefore
\begin{equation}\label{eq sys-closed}
\mathscr{G}_{e}:
\begin{cases}
\dot{x}=f(x,k(x+e),w),\\
z=h(x,w).\\
\end{cases} 
\end{equation}    
The measurement error $e$ is considered as an input of the system $\mathscr{G}_e$ in the next definition.
\begin{deff}\label{def ISS}
The ${\bf{C}}^1$ function $V: \mathbb{R}^n\mapsto \mathbb{R}_{\geq0}$ is an ISS Lyapunov function for system $\mathscr{G}_{e}$ defined in (\ref{eq sys-closed}) if there exist class $\mathcal{K}_\infty$ functions $\sigma$, $\sigma_i$, $\gamma_i$ ($i=1,2$) such that
\begin{equation}\label{ISS V eq}
\sigma_1(|\xi|)\leqslant V(\xi)\leqslant \sigma_2(|\xi|)
\end{equation}
holds for all $\xi\in\mathbb{R}^n$, and 
\begin{equation}\label{ISS eq}
\nabla V(\xi)\cdot f(\xi,k(\xi+\mu),w)\leqslant-\sigma(|\xi|)
\end{equation}
for any $\xi \in \mathbb{R}^n$, any $\mu \in \mathbb{R}^n$ and any $w\in \mathscr{W}_Q$ such that $|\xi|\geqslant \gamma_1(|\mu|)+\gamma_2(|w|)$.
\end{deff}
The next theorem suggests an equivalent condition to the above given inequality (\ref{ISS eq}). We will use this theorem later to develop our main theorem in Section \ref{section main result}.  
\begin{thm}\label{thm_equiv}
	The ${\bf{C}}^1$ function $V$ is an ISS Lyapunov function for system $\mathscr{G}_{e}$ if and only if (\ref{ISS V eq}) holds and there exist class $\mathcal{K}_{\infty}$ functions $\bar{\sigma}$ and $\beta_i$ ($i=1,2$) so that
	\begin{eqnarray}\label{last_help}
		\nabla V(\xi)\cdot f(\xi,k(\xi+\mu),w)\leqslant- \bar{\sigma}(|\xi|)+\beta_1(|\mu|)+\beta_2(|w|)
	\end{eqnarray}
	for any $\xi \in \mathbb{R}^n$, any $\mu \in \mathbb{R}^n$ and any $w\in \mathscr{W}_Q$.
\end{thm}
Definition \ref{def ISS} provides a characterization of the notion of input-to-state stability, rather than the ISS definition, using Lyapunov-like conditions. 
Next theorem shows that these conditions are necessary and sufficient for input-to-state stability.
\begin{thm}\label{thm ISS}
The closed loop system $\mathscr{G}_{e}$ defined in (\ref{eq sys-closed}) is ISS with respect to measurement error $e$ and disturbance $w$ iff there exists an ISS Lyapunov function $V$ satisfying (\ref{ISS V eq}), (\ref{ISS eq}).
\end{thm}

\begin{proof}
The proof follows from Theorem \ref{thm_equiv} and (\cite{ISS}, Theorem 1). 
\end{proof}
	\begin{rmk}\label{rmk_unify}
			Later in Section \ref{section main result} our study will focus on the systems with disturbances norm bounded by some function of state , {\it i.e.,} $|w(t)|\leqslant\gamma_3(|x(t)|)$. 
			This assumption seems to be implied in Definition \ref{def ISS} as condition (\ref{ISS eq}) is valid for $\gamma_2(|w(t)|)\leqslant|x(t)|-\gamma_1(|e(t)|)$. Thus to prevent any possible redundancy of these conditions, we will unify them later in section \ref{section main result}. 
	\end{rmk}

\subsection{Problem Setup}
To state our problem we shall need to define a continuous-time version of system $\mathscr{G}_{e}$ defined in (\ref{eq sys-closed}) by assuming measurement error to be zero all the time. This system will be referred as $\mathscr{G}_{c}$ throughout the rest of this note.
Now assume the existence of a positive definite ${\bf{C}}^1$ function $V$ and a ${\bf{C}}^0$ function $k:\mathbb{R}^n\mapsto \mathbb{R}^m$ such that $H_{\Gamma}(V,k(x))\leqslant0$, {\it i.e.}, the state feedback control law $u=k(x)$ renders the continuous-time system $\mathscr{G}_c$ finite gain locally $\mathcal{L}_2$-stable with zero bias and ${\|\mathscr{G}_c\|}_{\mathcal{L}_2}\leqslant {\Gamma}$. We also assume the implementation of the control law to be performed in an event-based scheme in which an event detector decides when to update the control signal. As a consequence, the actuator receives an updated control signal at triggering instants $\{t_i|i\in \mathbb{Z}_{\geq 0}\}$, at which an event condition is satisfied. The first sampling instant can always be assumed to coincide with initial time $t_0$.  
A zero order hold device serves to maintain the controller signal constant between two successive sampling instants. Thus, between time instants $t_i$ and $t_{i+1}$, the controller signal is $k(x(t_i))$ and remains unchanged.
This enables us to define the measurement error $e(t)$ as the difference between the current value of state at the event detector, $x(t)$, and the last triggered value of state, $x(t_i)$, {\it i.e.}, 
\begin{equation}\label{error}
e(t)=x(t_i)-x(t),~~~t\in[t_i,t_{i+1}).
\end{equation}
It follows that the measurement error is zero at each sampling instants and its value is continuously monitored to check a triggering condition which, as we will see later, sets an upper bound on the norm of admissible measurement error.
Once the triggering condition holds, the system sends an updated signal to the actuator and resets the measurement error to zero. 

In \cite{tabuada} it is shown that in presence of an execution rule that restricts the measurement error to satisfy
\begin{equation}\label{trig old}
\beta_1(|e|)\leqslant {c}\bar{\sigma}(|x|),
\end{equation}
where ${c}\in(0,1)$, and if there exists an ISS Lyapunov function $V$ so that
\begin{equation}
\nabla V(x)\cdot f(x,k(x+e),0)\leqslant -\bar{\sigma}(|x|)+\beta_1(|e|),
\end{equation}
the system $\mathscr{G}_e$ with zero-input is globally asymptotically stable.

In general, the aforementioned triggering mechanism (\ref{trig old}) guarantees closed loop stability. However, it is by no means clear how it affects the {\it input/output} performance of the system. More specifically, in this paper, we are concerned with finite gain $\mathcal{L}_2$-stability performance. The purpose of this paper is then to present an input-output stability analysis of event-based systems. Departing from the event condition offered in \cite{tabuada}, we propose a condition which guarantees the finite gain local $\mathcal{L}_2$-stability of the system.

\section{$\mathcal{L}_2$-gain Performance of Event Triggered Nonlinear Systems}\label{section main result}  
In this section we present a novel event-triggering rule that ensures finite gain local $\mathcal{L}_2$-stability of the event-based system $\mathscr{G}_e$. The design of such a sampling rule is based on the following assumptions.
\begin{assumption}\label{assumption2}
There exist a positive definite ${\bf{C}}^1$ function $W$ and some $Q\in\mathbb{R}_{> 0}$ such that 
\begin{equation}\label{gain}
H_{\Gamma}(W,k(x))\leqslant 0,
\end{equation}
for all $w\in\mathscr{W}_Q$, where $\mathscr{W}_Q$ is defined in (\ref{def_W_Q}).
\end{assumption}
\begin{assumption}\label{assumption1}
		There exis a radially unbounded positive definite ${\bf{C}}^1$ function $V$ and class $\mathcal{K}_\infty$ functions $\bar{\sigma}$, $\beta_1$ satisfying
		\begin{eqnarray}{\label{assump_eq_V_dot}}
			\nabla V(\xi) \cdot f(\xi,k(\xi+\mu),w)\leqslant -\bar{\sigma}(|\xi|)+\beta_1(|\mu|)
		\end{eqnarray}  	
		for any $\xi \in \mathbb{R}^n$, any $\mu \in \mathbb{R}^n$ and any $w\in \mathscr{W}_Q$.
\end{assumption}
Recalling Theorem \ref{thm gain}, condition (\ref{gain}) ensures that the continuous-time system $\mathscr{G}_c$ is finite gain locally $\mathcal{L}_2$-stable with zero bias and has ${\|\mathscr{G}_{c}\|}_{\mathcal{L}_2}\leqslant\Gamma$.	
The following lemma describes the connection between Assumption \ref{assumption1} and the previously defined ISS concept. Indeed, we show that this assumption can be used to deal with unmodeled parameter uncertainties.
\begin{lem}\label{lem unify}
	(a) Assumption \ref{assumption1} holds if and only if there exists a radially unbounded positive definite ${\bf{C}}^1$ function $V$ satisfying
	$\nabla V(\xi) \cdot f(\xi,k(\xi+\mu),w)\leqslant -{\sigma}(|\xi|)$ for any $\xi \in \mathbb{R}^n$, any $\mu \in \mathbb{R}^n$ and any $w\in \mathscr{W}_Q$ such that $|\xi|\geqslant \gamma(|\mu|)$ for some $\sigma,\gamma\in\mathcal{K}_\infty$. (b) The later condition is satisfied when for any $w\in\mathscr{W}_Q$ the followings hold:
		(I) $V$ is an ISS Lyapunov function for the system $\mathscr{G}_{e}$,
			(II) there exist solutions $\gamma_3$, $\gamma_4\in\mathcal{K}_{\infty}$ to the inequality \begin{eqnarray}\label{cond_id}
			\gamma_4\circ(\gamma_{id}-\gamma_2\circ\gamma_3)(r)\geqslant r,
		\end{eqnarray}
		for all $r\in\mathbb{R}_{\geq0}$, where $\gamma_{id}$ is the identity function and $\gamma_2\in\mathcal{K}_{\infty}$ is defined in Definition \ref{def ISS},
		(III) disturbance is bounded through
		\begin{equation}\label{bound_on_w}
			|w(t)|\leqslant\gamma_3(|x(t)|)
			\end{equation}
		for all $t\in\mathbb{R}_{\geq0}$ where $x$ denotes the state of system $\mathscr{G}_{e}$ defined in (\ref{eq sys-closed}).		
\end{lem}
Note that condition (\ref{cond_id}) is similar to $\delta$-admissible {\color{black}perturbation} provided in (\cite{dig_self_trig}, Definition 2).

We will need the following technical lemma to prove our main result. This lemma sets the stage for the design of the triggering condition required to achieve disturbance attenuation bound $\Gamma$ for the event-based system.
\begin{lem}\label{key_lem}
	Assumption \ref{assumption1} holds if and only if there exist a radially unbounded positive definite ${\bf{C}}^1$ function $V$ and class $\mathcal{K}_\infty$ functions $\hat{\sigma}$, $\sigma_0$, $\beta_0$, $\psi$, $\bar{\beta}_1$ and some $c\in(0,1)$ satisfying
	\begin{eqnarray}\label{lem_eq}
	\nabla V(\xi)\cdot f(\xi,k(\xi+\mu),w) \leqslant -\hat{\sigma}(|\xi|)-\sigma_0(|\xi|)\beta_0(|\mu|) 
	\end{eqnarray}
	for any $\xi \in \mathbb{R}^n$, any $\mu \in \mathbb{R}^n$ and any $w\in \mathscr{W}_Q$ such that $c\psi (|\xi|)\geqslant \bar{\beta}_1(|\mu|)$.
\end{lem}

\textit{Triggering Condition:} Let $t_i$, $i\in\mathbb{Z}_{\geq 0}$, be the most recent sampling instant, the control signal is {\color{black}updated again} at $t_{i+1}$ defined by the following rule:
\begin{eqnarray}\label{execution rule}
\begin{array}{l}
\hspace*{-0.5em}
t_{i+1}^-=\inf\Big{\{}t\in\mathbb{R}_{\geq 0}|~t>t_i \bigwedge \bar{\beta}_1(|e(t)|)\geqslant c\psi(|x(t)|) \Big{\}},\hspace*{-0.2em}
\end{array}
\end{eqnarray}
where $c\in(0,1)$ and $\psi$, $\bar{\beta}_1$ are defined as 
\begin{eqnarray}\label{key_eq}
\psi(r)\doteq \frac{\bar{\sigma}(r)}{1+\sigma_0(r)},~\bar{\beta}_1(r)\doteq \smash{\displaystyle\max}{\{\beta_1(r),\beta_0(r)\}}.
\end{eqnarray}
for $\sigma_0(r)=L_fL_k\sigma_3(r)$ and $\beta_0(r)=r$ with $L_f$, $L_k$ defined in Remark \ref{rmk lip f}. Note that we assume that the update of the control task is done at $t_{i+1}$, shortly after the given inequality in (\ref{execution rule}) is satisfied at $t_{i+1}^-$
The following theorem states that if the continuous-time system has some local $\mathcal{L}_2$-gain property, it is always possible to guarantee the same disturbance attenuation level for the event-based system by applying the above triggering mechanism.
\begin{thm}\label{thm main}
Let us consider Assumptions \ref{assumption2}, \ref{assumption1} and the following conditions:
\begin{enumerate}[label=\roman*{$)$}]
\item [(i)] \label{cond1} $|\nabla W(x)|\leqslant \sigma_3(|x|)$ for some class $\mathcal{K}_\infty$ function $\sigma_3$, locally Lipschitz-continuous in $\mathbb{R}_{\geq0}$,
\item [(ii)] \label{cond3} ${\bar{\sigma}}^{-1}$, $\beta_1$, $\gamma_3$ are locally Lipschitz-continuous in $\mathbb{R}_{\geq0}$ \footnote{\label{note1} This condition can be relaxed in the proof of Theorem \ref{thm main}, however, is needed in the proof of Theorem \ref{thm main execution time}.},
\item [(iii)] \label{cond4} $k$ and $f$ are locally Lipschitz-continuous in $\mathbb{R}^n$ and $\mathbb{R}^n\times \mathbb{R}^m\times \mathbb{R}^q$, respectively \footnotemark[\getrefnumber{note1}].
\end{enumerate}
Then the system $\mathscr{G}_{e}$ driven from initial conditions $x_0\in\mathscr{X}_0$, defined in (\ref{initial_cond_set}), is finite gain locally $\mathcal{L}_2$-stable with zero bias and has ${\|\mathscr{G}_{e}\|}_{\mathcal{L}_2}\leqslant \Gamma$ if the control signal is executed under rule (\ref{execution rule}).
\end{thm}

It is worth remarking that Theorem \ref{thm main} is stated in local form. Note that condition (\ref{bound_on_w}) which restricts $w$ to be norm bounded by some Lipschitz-continuous function of state, plays an essential role in satisfying Assumption \ref{assumption1}. This assumption is not consistent with classical input-output stability notion that requires $w$ to be {\it any} perturbation in $\mathcal{L}_2(\mathbb{R}_{\geq 0})$. Thus it remains to define $Q$ such that for any given initial conditions in $\mathscr{X}_0$, $w$ is guaranteed to be in the set $\mathscr{W}_Q$. Condition (\ref{bound_on_w}) is a key tool to define such an admissible inputs set. Indeed, later in view of Lemma \ref{pro IC}, condition (\ref{bound_on_w}) and Lipschitz-continuity of $\gamma_3$ with Lipschitz constant $L_{\gamma_3}$ defined in Remark \ref{rmk lip f}, one can choose $Q=L_{\gamma_3}\bar{\varepsilon}$.
\begin{rmk}\label{remark_represent w}
The assumed dependence of $\gamma_3$ on the state of the system in (\ref{bound_on_w}) is a generalization of the assumption of state-dependent disturbance made in (\cite{L_2_selftriggered}, Assumption 6.1). Indeed, the Assumption 6.1 in \cite{L_2_selftriggered} can be extracted from (\ref{bound_on_w}) by choosing $\gamma_3$ to be a linear function of state, {\it i.e.}, $\gamma_3(|x|)=c_0|x|$, for all $x\in\mathbb{R}^n$ and some $c_0\in\mathbb{R}_{>0}$. This generalization has to be considered more carefully as it gives more flexibility in choosing function $\gamma_2$ in (\ref{cond_id}), {\it e.g.}, for $\gamma_2(r)=\sqrt{r}$, (\ref{cond_id}) does not provide any solution for possible linear functions $\gamma_3$. However, it is not difficult to verify that the solution to this inequality exists assuming $\gamma_3$ to be locally Lipschitz-continuous in $\mathbb{R}_{\geq 0}$.
\end{rmk}
\begin{rmk}
Using the same discussion as in (\cite{L_2_selftriggered}, Remark 6.2), it is more precise to state condition (\ref{bound_on_w}) as $|w(t,x(t))|\leqslant \gamma_3(|x(t)|)$ for all $t\in \mathbb{R}_{\geq0}$ to emphasize the state dependence of exogenous disturbance. To simplify our notation, we write $w(t)$ instead of $w(t,x(t))$ throughout the rest of the paper. 
\end{rmk}
\begin{rmk}\label{rmk w=0}
The triggering condition (\ref{trig old}) proposed in \cite{tabuada} can be extracted from the one we proposed in (\ref{execution rule}). Indeed, between consecutive sampling instants, (\ref{execution rule}) suggests
\setlength{\arraycolsep}{0.0em}
\begin{eqnarray}\label{comparison}
{c}\bar{\sigma}(|x|)&{}\geqslant{}& \smash{\displaystyle\max}{\{\beta_1(|e|),\beta_0(|e|)\}}(1+\sigma_0(|x|))\nonumber \\ &{}\geqslant{}& \beta_1(|e|)+\beta_0(|e|)\sigma_0(|x|)
\end{eqnarray}
and hence we conclude that $\beta_1(|e|)\leqslant {c} \bar{\sigma}(|x|)$. 
This consequence simply suggests that under the conditions assumed in this paper, in order to preserve system performance (in $\mathcal{L}_2$ sense) along with asymptotic stability provided in \cite{tabuada}, a more conservative execution rule than the one proposed in \cite{tabuada} is needed.
\end{rmk}
Our next Lemma shows that the state of the event-based system $\mathscr{G}_e$ is constrained to some compact set. The result is fundamental in the rest of this section.
\begin{lem}\label{pro IC}
	Under the assumptions of Theorem \ref{thm main}, $\mathscr{X}\doteq \{r\in\mathbb{R}^n|~|r|\leqslant \bar{\varepsilon}\}$ for $\bar{\varepsilon}=\sigma_1^{-1}(\sigma_2(\varepsilon))$ is a positive invariant set for the trajectories of system $\mathscr{G}_e$ driven from any $x_0\in\mathscr{X}_0$.
\end{lem}
\begin{rmk}\label{rmk lip f}
We now show how this analysis can be applied to find upperbounds on the norm of $\dot{x}$, something needed later to exclude Zeno-behaviour for the system $\mathscr{G}_e$. 
Lemma \ref{pro IC} suggests that $x(t)$ remains in the compact set $\mathscr{X}$ for all $t\in\mathbb{R}_{\geq0}$. Moreover, in view of definition of $e$ given in (\ref{error}) we have $|e(t)|\leqslant 2\bar{\varepsilon}$ for all $t\in\mathbb{R}_{\geq 0}$. Thus we may conclude that $e(t)\in {\mathscr{X}}_e\doteq \{r\in\mathbb{R}^n|r/2\in\mathscr{X}\}$ for all $t\in\mathbb{R}_{\geq 0}$. Also the control signal $u=k(x+e)$ does not leave the compact set ${\mathscr{X}}_u\doteq \{r\in\mathbb{R}^n|r/k\in\mathscr{X}\}$ since $|u(t)|\leqslant k|x(t_i)|= k\bar{\varepsilon}$ for all $t\in\mathbb{R}_{\geq 0}$.
Now consider compact sets $\mathscr{B}_x \subset \mathscr{X}$ and $\mathscr{B}_e \subset {\mathscr{X}}_e$. In the view of Lipschitz-continuity of function $k$, one can define the compact set $\mathscr{B}_u\subset {\mathscr{X}}_u$ of all points $u\in \mathbb{R}^m$ satisfying $|u|\leqslant |k(x+e)|$ for all $x\in \mathscr{B}_x$ and $e\in \mathscr{B}_e$. Similarly, we can define the compact set $\mathscr{B}_w\subset \mathscr{W}_Q$ containing all  points $w\in\mathbb{R}^q$ satisfying $|w|\leqslant \gamma_3(|x|)$ for all $x\in \mathscr{B}_x$.
Now using the Lipschitz-continuity of function $f$ with respect to $(x^T~u^T~w^T)^T$ in compact set $\mathscr{B}_x \times \mathscr{B}_u \times \mathscr{B}_w$ with $L_f$ is the Lipschitz constant of the function $f$ on $\mathscr{X} \times {\mathscr{X}}_u \times \mathscr{W}_Q$ and applying triangle inequality $|f(x,u,w)-f(\tilde{x},\tilde{u},\tilde{w})| \leqslant
|f(x,u,w)-f({x},\tilde{u},{w})|+
|f(x,\tilde{u},w)-f(\tilde{x},\tilde{u},\tilde{w})|$, it is not difficult to confirm the Lipschitz-continuity of function $\bar{f}(x,e,w)\doteq  f(x,k(x+e),w)$ in any compact set $\mathscr{B}_x\times \mathscr{B}_e\times \mathscr{B}_w$ with Lipschitz constant $L_f(L_k+1)$.
It is also straight forward to check
\setlength{\arraycolsep}{0.0em}
\begin{eqnarray}
{}&&{}|\dot{x}|\leqslant L_f(L_k+1)|x|+L_f L_k|e|+L_f|w|, \label{dotx0} \\
{}&&{}|f(x,k(x+e),w)-f(x,k(x),w)|\leqslant L_fL_k|e|,  \label{rmk lip eq4}
\end{eqnarray}
that will be used further. Also inequality (\ref{dotx0}) in view of condition (\ref{bound_on_w}) in Theorem \ref{thm main} and Lipschitz-continuity of $\gamma_3$ in the compact set $\{r\in\mathbb{R}_{\geq 0}| r\leqslant \max_{x\in\mathscr{B}_x}|x|\}$ with Lipschitz constant $L_{\gamma_3}$ (defined on $[0,\bar{\varepsilon}]$), reads as
\begin{equation}\label{dotx}
\begin{gathered}
|\dot{x}|\leqslant L_f(L_k+L_{\gamma_3}+1)|x|+L_fL_k|e|.
\end{gathered}
\end{equation}
\end{rmk}

In the next theorem we show that the sequence of triggering instants is a uniformly isolated set and hence there always exists a non-zero lower bound $\tau$ on the intersampling times. This feature guarantees the non-existence of accumulation points and is thus critical to the successful implementation of the proposed triggering mechanism.
\begin{thm}\label{thm main execution time}
If the hypotheses of Theorem \ref{thm main} hold, the inter sampling periods are lower bounded by some $\tau\in\mathbb{R}_{>0}$, {\it i.e.}, $t_i\geqslant t_{i-1}+\tau$ for all $i\in\mathbb{Z}_{>0}$.
\end{thm}
Proof of Theorem \ref{thm main execution time} relies on Properties \ref{pro lip}-\ref{pro beta} outlined below and is given in the Appendix.
\begin{property}\label{pro lip}
Function $\psi^{-1}$ defined in (\ref{key_eq}) is Lipschitz-continuous in any compact set $\mathscr{D}_x\subset\mathbb{R}_{\geq0}$.
\end{property}
\begin{property}\label{pro beta}
	The function $\bar{\beta}_1(r)$ defined in (\ref{key_eq}) is of class $\mathcal{K}_\infty$ and locally Lipschitz-continuous in $\mathbb{R}_{\geq0}$. Also if the Lipschitz constant of function $\beta_1$ is $L_{\beta_1}$ on some compact set $\mathscr{D}_e\subset \mathbb{R}_{\geq 0}$, then $L_{\bar{\beta}_1}=\max\{L_{\beta_1},1\}$ is the Lipschitz constant of $\bar{\beta}_1$ on this set. 
\end{property}

Proof of Theorem \ref{thm main execution time} implies that $\tau$ is a function of $L_f$, $L_k$, $L_{\gamma_3}$. Applying Lemma \ref{pro IC} and Remark \ref{rmk lip f}, we conclude that these {\color{black}constants are} defined on invariant sets and hence are valid for all initial conditions.

The proof of Property \ref{pro lip} suggests that function $\psi^{-1}$ is Lipschitz-continuous in $\mathscr{D}_x$ with Lipschitz constant $L_{\psi^{-1}}=(1+\sigma_0(\varepsilon_m))^2\big{\{}{L^{-1}_{\bar{\sigma}^{-1}}}-L_{\sigma_0}\bar{\sigma}(\varepsilon_m)\big{\}}^{-1}$, where $\varepsilon_m=\max_{r\in\mathscr{D}_x}\{r\}$ and $L_{\sigma_0}=L_fL_kL_{\sigma_3}$. To make sure $L_{\psi^{-1}}$ is positive, $\sigma_3$ which is the upper bound on the norm of Lyapunov function $W$, has to be chosen so that $L_fL_kL_{\sigma_3}L_{\bar{\sigma}^{-1}}\bar{\sigma}(\varepsilon_m)<1$. This condition depends on the set $\mathscr{D}_x$. In design procedure, however, one can choose $\sigma_3$ such that
\begin{equation}\label{lambda_bar_restriction}
L_fL_kL_{\sigma_3}L_{\bar{\sigma}^{-1}}\bar{\sigma}(\bar{\varepsilon})<1,
\end{equation}
where $\bar{\varepsilon}$ is defined in Lemma \ref{pro IC}. To see this, let us assume that system starts from initial condition $x(0)=x_0$. Then in view of Lemma \ref{pro IC}, we have $|x(t)|\leqslant \bar{\varepsilon}$ and hence for any compact set $\mathscr{D}_x\subseteq[0,\bar{\varepsilon}]$ we have $\varepsilon_m\leqslant \bar{\varepsilon}$. Thus since $\bar{\sigma}$ is a class $\mathcal{K}_{\infty}$ function, we will have ${\bar{\sigma}(\varepsilon_m)}\leqslant {\bar{\sigma}(\bar{\varepsilon})}$ and hence (\ref{lambda_bar_restriction}) ensures that $L_fL_kL_{\sigma_3}L_{\bar{\sigma}^{-1}}\bar{\sigma}(\varepsilon_m)<1$.

We finish our discussions in this section by showing the global asymptotic stability property for the event-based system $\mathscr{G}_e$ in the absence of disturbances.
\begin{cor}\label{cor w-Iss}
Under the assumptions of Theorem \ref{thm main}, the zero-input event-based system $\mathscr{G}_e$ has a global asymptotically stable point at $0\in\mathbb{R}^n$.
\end{cor}

Although the $\mathcal{L}_2$-stability results are provided locally, the above result is global. Recall that the local character of the results arises from the {\color{black}restrictions placed} on the input space. Thus, in the absence of disturbances, the result becomes global.
\section{Improving Average Sampling Frequency}\label{section improve inter-execution time}
In this section, we are concerned with the problem of decreasing the average sampling rate for the proposed triggering mechanism of Section \ref{section main result}. Note that limiting the number of triggerings over a time interval (when necessary) may be of higher importance than decreasing {\color{black}their number over the entire time span}.  For example, high data load on a communication channel over a finite time interval may result in undesirable effects such as data packet drop out and/or transmission delay. Therefore, instead of improving all inter-event times, we focus our study on controlling the number of samples over a time interval in which the triggering frequency may become critical.


Our solution consists of modifying the triggering condition (\ref{execution rule}) of section \ref{section main result} by adding an exponentially time decaying term to the right hand side. We show that following this idea, the event-based system enjoys the same $\mathcal{L}_p$-gain performance as in section \ref{section main result}, however, the zero-input system is stable in practical sense as opposed to asymptotically stable.

The results of this section can then be applied to limit high triggering during transient response. In a regulation problem, after an initial transition, the state remains near the equilibrium, possibly continuously excited by a disturbance or noise. Focusing on practical stability of such problems, where the state is required to enter an stability bound, it is reasonable to assume that the triggering frequency reduces when the transient response vanishes. Note that when the state is near the equilibrium, the control action is only required to keep the state within the desired bound. As a result, the system can be controlled with much less attention and hence the number of triggering instants drops significantly. Similar behaviour is expected when tackling regulation problems with non-persistent disturbances. In such case, while in transient both disturbance and the change in the state's norm affect sampling frequency, during steady state a lower triggering rate is expected due to the non-existence of disturbance.

We now state the main problem to be solved in this section. Note that we implicitly assume that the system experiences finite transition interval over which the sampling frequency exceeds a critical level. Without loss of generality, we assume that only one such interval exists. Generalization to several transition intervals is discussed later.
\\
\textit{Problem 1:} 
Modify the proposed triggering rule (\ref{execution rule}) so that while the resulting event-based system is finite gain locally $\mathcal{L}_2$-stable with the same disturbance rejection bound $\Gamma$, the average sampling frequency does not exceed $f_{cr}$ at least for
\begin{itemize}
	\item [(A)] a desired period of time, {\it i.e.,} $t\in[0,\bar{T}]$.
	\item [(B)] a desired number of triggerings, {\it i.e.,} $1\leqslant i \leqslant N$.
\end{itemize}
	

\subsection{Continuous Triggering Condition Scenario}\label{section_cont}
We begin our study of Problem 1-(A) by modifying rule (\ref{execution rule}) as
\begin{eqnarray}\label{new execution}
\begin{array}{l}
\hspace*{-0.5em}
	t_{i+1}^-=\inf\Big{\{}t\in\mathbb{R}_{\geq 0}|~t>t_i \bigwedge \bar{\beta}_1(|e(t)|)\geqslant c\tilde{\psi}(|x(t)|) \Big{\}},\hspace*{-0.3em}
\end{array}
\end{eqnarray}
where $\tilde{\psi}\doteq {\tilde{\sigma}(t,r)}/({1+\sigma_0(r)})$ and $\tilde{\sigma}$ is an exponentially time-decaying perturbation of function $\bar{\sigma}$ defined in Assumption \ref{assumption1}, {i.e.,}
\begin{equation}\label{cont_perturbation}
	\tilde{\sigma}(t,r)\doteq\bar{\sigma}(r)+\frac{\kappa}{c} e^{-\zeta t}.
\end{equation}
Also $\kappa$ and $\zeta$ are positive parameters to be designed. 

\subsubsection{Stability Analysis}\label{subsection L2 stability}
The following theorem shows that the time-decaying perturbation of function $\bar{\sigma}$ introduces a non-zero bias term (see Definition \ref{def gain})
but does not affect the  $\mathcal{L}_2$ bound with respect to the input.

\begin{thm}\label{thm practical0}
Under the hypotheses of Theorem \ref{thm main}, the system $\mathscr{G}_e$ is finite gain locally $\mathcal{L}_2$-stable and has ${\|\mathscr{G}_e\|}_{\mathcal{L}_2}\leqslant \Gamma$ if the control signal is updated under the execution rule (\ref{new execution}).
\end{thm}
\begin{rmk}
	It can be readily inferred from the proof of Theorem \ref{thm practical0} that the exponential time decaying term in (\ref{new execution}) does not affect the finite gain local $\mathcal{L}_2$-stability of the event-based system $\mathscr{G}_e$ as its integral from $0$ to \emph{any} $T\in\mathbb{R}_{>0}$ is finite, independent of $T$ and hence can be considered as the bias term $\eta$ in Definition \ref{def gain}.
\end{rmk}
\begin{cor}\label{cor w-Iss2}
Under the assumptions of Theorem \ref{thm practical0}, trajectories of the system $\mathscr{G}_e$ converge to $0\in\mathbb{R}^n$.
\end{cor}
\begin{rmk}
Note that Corollary \ref{cor w-Iss2} proves that trajectories converge to the origin, but \underline{does not} imply that the origin is asymptotically stable for the disturbance-free system.
Asymptotic stability does not follow from this corollary since the result falls short of proving \emph{stability} of the origin of the zero-input event-based system $\mathscr{G}_e$.
This situation may occur, for example, when the trajectories of the zero-input system that start from certain neighbourhood of the origin, diverge from origin temporarily, but finally converge to it. In such situations, the system may still be finite gain locally $\mathcal{L}_2$-stable, however, the zero-input system is not necessarily stable since there exist neighbourhoods of the origin such that any trajectory starting there, can not stay there {\color{black}forever}. This happens for system $\mathscr{G}_e$ under triggering condition (\ref{new execution}) as the proof of Corollary \ref{cor w-Iss2} suggests that in the absence of disturbances we have $\dot{V}(x)\geqslant 0$ for $|x(t)|\leqslant \bar{\sigma}^{-1}(\kappa e^{-\zeta t}/(1-c))$, {\it i.e.,} trajectories starting within this bound diverge from origin at first but finally converge as the area of positive $\dot{V}$ shrinks to zero.
\end{rmk}

The analysis, however, can be extended a bit further than the classical notion of stability. Indeed, we now show that in the absence of disturbances, the event-based system $\mathscr{G}_e$ is practically stable in the sense of following definition cited from \cite{practical_stability}:
\begin{deff}\label{practical stability}
Given $\varsigma>\rho\in \mathbb{R}_{\geq 0}$, the origin of the system $\dot{x}=f(x,t)$ is $(\varsigma \rightarrow \rho)$-stable if
	\begin{enumerate}[label*=\alph*)]
	\item [(a)] \label{1} for any $\epsilon > \rho$ there exists $\delta(\epsilon)\in\mathbb{R}_{>0}$ such that if $|x_0|\leqslant \delta(\epsilon)$, then $|x(t)|<\epsilon$ for all $t\in\mathbb{R}_{\geq 0}$,
	\item [(b)] \label{2} for a given $r\in(0,\varsigma)$ there exists a finite $\upsilon(r)\in\mathbb{R}_{>0}$ such that if $|x_0|\leqslant r$, then $|x(t)|<\upsilon(r)$ for all $t\in\mathbb{R}_{\geq 0}$,
	\item [(c)] \label{3} for a given $r\in(0,\varsigma)$ and $\epsilon > \rho$ there exists a finite $T(r,\epsilon)\in\mathbb{R}_{>0}$ such that if $|x_0|\leqslant r$, then $|x(t)|<\epsilon$ for all $t\geqslant T(r,\epsilon)$.
\end{enumerate}
\end{deff}
If we set $\varsigma=\infty$ and $\rho=0$ in the above definition, we obtain the familiar uniform global asymptotic stability, (\cite{practical_stability}, Remark 2.1). it is worth mentioning that the stability in the above-mentioned $\rho$-practical sense guarantees convergence of trajectories of the system $\dot{x}=f(x,t)$ to the set $\{x\in\mathbb{R}^n|~|x|\leqslant \rho\}$ through condition \ref{3} in Definition \ref{practical stability}. The converse, however, is not generally true.
\begin{thm}\label{thm practical}
Under the assumptions of Theorem \ref{thm practical0}, the zero-input event-based system $\mathscr{G}_e$ is $(\infty \rightarrow \bar{\sigma}^{-1}({\kappa}/(1-c)))$-stable.
\end{thm}

\subsubsection{Inter-Event Lower Bound Comparison}\label{subsection lower bound increase} 
We recall from Theorem \ref{thm main execution time} that the lower bound on intersampling periods of event-based system $\mathscr{G}_e$ under execution rule (\ref{execution rule}) is $\tau$ and given in (\ref{lower bound}). Also by $\tau_1$ we denote the lower bound on intersampling periods of this system under execution rule (\ref{new execution}). We show that one can design parameters $\kappa$ and $\zeta$ in (\ref{new execution}) such that for a given $\bar{T}>0$ and ${\tau}^*>f_{cr}^{-1}$, we have $\tau_1\geqslant\tau+\tau^*$ at least for $t\in [0,\bar{T}]$. This guarantees that the average sampling frequency is less than $f_{cr}$ for $t\in[0,\bar{T}]$. To this end, defining $\bar{\kappa}={\kappa}/({1+\sigma_0(\bar{\varepsilon})})$, we assume the updation of the control task is decided based on the following event condition
\begin{equation}\label{execut_more_restric1}
	\bar{\beta}_1(|e|)\geqslant c\psi(|x|)+\bar{\kappa}e^{-\zeta t}
\end{equation}
which is more conservative than the one proposed in (\ref{new execution}) and hence gives a lower bound on $\tau_1$. Let $L_{\psi^{-1}}$ and $L_{\bar{\beta}_1}$ be the Lipschitz constants of functions $\psi^{-1}$ and $\bar{\beta}_1$, respectively. A more conservative triggering condition than (\ref{execut_more_restric1}) can be obtained from  $L_{\bar{\beta}_1}|e|\geqslant {c}  L_{\psi^{-1}}^{-1}|x|+\bar{\kappa}e^{-\zeta t}$. In fact, if this condition is not satisfied, we have
\begin{equation}\label{cont analysis 1}
\bar{\beta}_1(|e|)\leqslant L_{\bar{\beta}_1}|e|< {c}  L_{\psi^{-1}}^{-1}|x|+\bar{\kappa}e^{-\zeta t}\leqslant {c} \psi(|x|)+\bar{\kappa}e^{-\zeta t}
\end{equation}
and hence (\ref{execut_more_restric1}) will not be satisfied too. This triggering condition restricts measurement error $e$ to satisfy 
\begin{equation}\label{cont analysis 2}
{c} L_{\psi^{-1}}^{-1} |x|(\hat{L}\frac{|e|}{|x|}-1) \leqslant \bar{\kappa}e^{-\zeta t},
\end{equation}
where $\hat{L}={c}^{-1} L_{\psi^{-1}} L_{\bar{\beta}_1}$. We remark that $\hat{L}\geqslant \bar{L}$, where $\bar{L}$ is the Lipschitz constant of function ${\psi}^{-1}({\bar{\beta}_1}/{{c}})$. From the proof of Theorem \ref{thm main execution time} it follows that ${|e|}/{|x|}\geqslant  {1}/{\bar{L}}$ shortly before the execution instant $t_i$
and hence we have $\hat{L}{|e(t_i^-)|}/{|x(t_i^-)|}>1$. We can even express the triggering condition more conservatively, by virtue of Lemma \ref{pro IC}, so that the control signal is updated at sampling instant $t_i$ when the following condition is satisfied
\begin{equation}\label{cont analysis 3}
{c}\bar{\varepsilon} L_{\psi^{-1}}^{-1} (\hat{L}\frac{|e(t_i^-)|}{|x(t_i^-)|}-1) \geqslant \bar{\kappa}e^{-\zeta t_i^-}.
\end{equation}
We now define $L^*$ so that $y(\tau+\tau^*) =1/L^*$ where $y$ is the solution to (\ref{differnetial eq}). Thus our aim is to design ${\kappa}$ and $\zeta$ such that the solution ${|e(t_i^-)|}/{|x(t_i^-)|}$ to inequality (\ref{cont analysis 3}) satisfy $L^*|e(t_i^-)|\geqslant |x(t_i)|$ for all execution instants $t_i\leqslant \bar{T}$, $i\in\mathbb{R}_{\geq0}$, {\it i.e.}, until $t=\bar{T}$ the intersampling intervals are lower bounded by the solution $\tau_1$ of $y(\tau_1)={1}/{L^*}$. This means that the lower bound on inter-event times increases to $\tau_1\geqslant \tau+\tau^*$ at least until instant $t=\bar{T}$. Finally it suffices to choose $\kappa$ and $\zeta$ so that
\begin{equation}
{\kappa}={c}\bar{\varepsilon} L_{\psi^{-1}}^{-1} (\frac{\hat{L}}{L^*}-1) ({1+\sigma_0(\bar{\varepsilon})})e^{\zeta \bar{T}}.
\end{equation}
Then the lower bounds on intersampling periods are the solutions $\tau_1$ and $\tau$ to
\begin{equation}
\begin{cases}
y(\tau_1)=\frac{1}{L^*},&\text{for } 0\leqslant t \leqslant \bar{T}\\
y(\tau)=\frac{1}{\bar{L}},&\text{for } t> \bar{T}\\
\end{cases}
\end{equation}  
\begin{rmk}
Our result in section \ref{subsection lower bound increase} is far more general than that of (\cite{tabuada}, Theorem $\textrm{III}.1$, when the delay between state measurement and actuator updating is nonzero). In \cite{tabuada}, it is shown that the lower bound on intersampling times, $\tau$, can be extended (due to the time required to read state measurement, compute the control signal and update actuators) to the solution $\tau^{\prime}$ of $y(\tau^{\prime})={1}/{{\bar{L}}^{\prime}}$, where ${\bar{L}}^{\prime}$ is the Lipschitz constant of function ${\psi}^{-1} ({\bar{\beta}_1}/{{c}^{\prime}})$ on compact set $\mathscr{D}_e$ defined in the proof of Theorem \ref{thm main execution time}, where ${c}^{\prime}\in({c},1)$. Following this approach, the lower bound on intersampling intervals is restricted (through the upper bound limit on ${c}^{\prime}$) to $\tau^{\prime}<\tau_{\max}$, where $\tau_{\max}$ is the solution to $y(\tau_{\max})={1}/{{\bar{L}}_{\min}}$ with ${\bar{L}}_{\min}$ as the Lipschitz constant of function ${\psi}^{-1}({\bar{\beta}_1})$. This limitation, however, is relaxed in our proposed method by introducing the exponentially decaying term $\kappa e^{-\zeta t}$ which allows taking $L^*$ smaller than ${\bar{L}}_{\min}$.
\end{rmk}
\subsection{Discrete Triggering Condition Scenario}\label{subsection discrete}
In this section we address Problem 1-(B). In section \ref{section_cont} we showed that an exponentially time decaying term added to execution rule (\ref{execution rule}) enables us to affect average sampling frequency arbitrarily at least for the period $[0,\bar{T}]$. Here, we address the problem of improving the average sampling frequency for the first $N$ iterative triggerings of the control task using a discrete version of the triggering condition (\ref{new execution}). Let $i\in \mathbb{Z}_{>0}$ denote the number of triggerings completed up until time $t$ assuming the first triggering occurs at $t_0=0$. As a consequence, $t_i$ and $t_{i-1}$ denote the upcoming and the most recent execution instants, respectively. We denote by $t_i^{\prime}>t_{i-1}$, $i\in \mathbb{Z}_{>0}$, just a moment after the following so called discrete event condition holds 
\begin{equation}\label{discrete execution}
	\begin{array}{l}\hspace*{-1em}
	{t'_i}^-=\inf\Big{\{}t\in\mathbb{R}_{\geq 0}|~t>t_{i-1} \bigwedge \bar{\beta}_1(|e(t)|)\geqslant c\check{\psi}(|x(t)|) \Big{\}},\hspace*{-0.3em}
	\end{array}
\end{equation}
where $\check{\psi}\doteq {\check{\sigma}(t,r)}/({1+\sigma_0(r)})$ and $\check{\sigma}$ is a discrete decaying perturbation of $\bar{\sigma}$ defined in Assumption \ref{assumption1} defined as
\begin{equation}\label{disc_perturbation}
\check{\sigma}(t,r)\doteq\bar{\sigma}(r)+\frac{\hat{\kappa} e^{\theta^i}}{c i!},
\end{equation}
where $\hat{\kappa}$ and $\theta$ are positive parameters to be designed. 
We refer to (\ref{discrete execution}) as the discrete triggering condition as it depends on index $i$ which changes non-continuously between successive triggerings. 

Now suppose that the $i$-th execution of the control task happens at
\begin{equation}\label{2 time}
t_i=\smash{\displaystyle\min}\{t_{i-1}+\Delta,t^{\prime}_i\},
\end{equation}
where $\Delta\in\mathbb{R}_{>0}$ is an upper bound on intersampling intervals. The following theorem then states that discrete decaying perturbation of function $\bar{\sigma}$ given in (\ref{disc_perturbation}) satisfies the same local $\mathcal{L}_2$-gain  bound for the event-based system.
\begin{thm}\label{thm practical2}
Under the hypotheses of Theorem \ref{thm main}, the system $\mathscr{G}_e$ is finite gain $\mathcal{L}_2$-stable and has ${\|\mathscr{G}_e\|}_{\mathcal{L}_2}\leqslant {\Gamma}$ if the control signal is updated at triggering instants $\{t_i|i\in\mathbb{Z}_{>0}\}$ defined in (\ref{2 time}).
\end{thm}
\begin{rmk}
	The $\Delta$ term in (\ref{2 time}) imposes an upper bound on inter-event times. This restriction on intersampling intervals is necessary as it confirms the finiteness of the bias term in (\ref{dissipative discrete}).
\end{rmk}
\begin{cor}\label{cor w-Iss discrete}
Under the assumptions of Theorem \ref{thm practical2}, trajectories of the event-based system $\mathscr{G}_e$ converge to $0\in\mathbb{R}^n$.
\end{cor}
\begin{thm}\label{thm_disc}
Under the assumptions of Theorem \ref{thm practical2} and in absence of disturbances, the event-based system $\mathscr{G}_e$ is $(\infty \rightarrow \bar{\sigma}^{-1}({\hat{\kappa}_\theta}/(1-c)))$-stable, where ${\hat{\kappa}_\theta}\doteq {\hat{\kappa}\theta^{\lfloor \theta \rfloor}}/{{\lfloor \theta \rfloor}!}$.
\end{thm}
In the rest of this section, we provide a discrete counterpart to the analysis given in section \ref{subsection lower bound increase}. Indeed, we design $\hat{\kappa}$ and $\theta$ in (\ref{discrete execution}) so that given some $N\in\mathbb{Z}_{>0}$ and ${\tau}^*\in\mathbb{R}_{>0}$, we have $\tau_2\geqslant\tau+\tau^*$ at least for $t\in [0,\bar{T}]$, where $\tau_2$ denotes the lower bound on intersampling periods of system $\mathscr{G}_e$ under execution rule (\ref{discrete execution}).

Choosing $\Delta>\tau+\tau^*$ in (\ref{2 time}), it remains to consider the case where $t_i=t^{\prime}_i$ and hence $t_i$ satisfy the triggering condition (\ref{discrete execution}). Even a more conservative event condition can be obtained if the $i$-th execution of control task is fulfilled when the following holds
\begin{equation}
	\bar{\beta}_1(|e|)\geqslant c\psi(|x|)+\tilde{\kappa}\frac{{\theta^i}}{i!},
\end{equation}
where $\tilde{\kappa}={\hat{\kappa}}/({1+\sigma_0(\bar{\varepsilon})})$. Now using the same procedures as (\ref{cont analysis 1}) and (\ref{cont analysis 2}) were derived, we obtain a discrete version of event condition (\ref{cont analysis 3}):
\begin{equation}\label{disc analysis 3}
{c}\bar{\varepsilon} L_{\psi^{-1}}^{-1} (\hat{L}\frac{|e(t_i^-)|}{|x(t_i^-)|}-1) \geqslant \tilde{\kappa}\frac{\theta^i}{i!}.
\end{equation}
Our goal is to design $\hat{\kappa}$ and $\theta$ such that the solution ${|e(t_i^-)|}/{|x(t_i^-)|}$ to the above inequality satisfies $L^*|e(t_i)|\geqslant |x(t_i)|$ for the first $N$ triggerings, where $L^*$ is defined such that $y(\tau+\tau^*)=1/L^*$ and $y$ is the solution to (\ref{differnetial eq}). 
We now consider two cases.
\\Case $1$: If $1\leqslant N \leqslant N_\theta\doteq  \smash{\displaystyle\operatorname\max}\{i|{\theta^i}/{i !}\geqslant \theta \}$ we have $\min_{1\leqslant i\leqslant N}\{{\theta^i}/{i !}\}=\theta$, {\it i.e.}, the discrete function ${\theta^i}/{i !}$ takes its minimum value at $i=1$, and hence we can choose $\hat{\kappa}$ and $\theta$ so that
\begin{equation}
\hat{\kappa}=\frac{{c}\bar{\varepsilon}}{\theta} L_{\psi^{-1}}^{-1} (\frac{\hat{L}}{L^*}-1) ({1+\sigma_0(\bar{\varepsilon})}).
\end{equation}
That is, for any $1\leqslant N \leqslant N_\theta$, the first $N_\theta$ inter-event intervals are lower bounded by the solution $\tau_2$ of $y(\tau_2)={1}/{L^*}$. 
\\Case $2$: For $N>N_\theta$ we have $\min_{1\leqslant i\leqslant N}\{{\theta^i}/{i !}\}={\theta^N}/{N !}$ and we can pick $\hat{\kappa}$ and $\theta$ such that
\begin{equation}
\hat{\kappa}={c}\bar{\varepsilon} L_{\psi^{-1}}^{-1} (\frac{\hat{L}}{L^*}-1)({1+\sigma_0(\bar{\varepsilon})}) \frac{N !}{\theta^N},
\end{equation}
{\it i.e.}, for the first $N$ samplings, the inter-event times are lower bounded by the solution $\tau_2$ of $y(\tau_2)={1}/{L^*}$. Therefore, the lower bounds on intersampling periods are the solutions $\tau_2$ and $\tau$ to
\begin{equation}
\begin{cases}
\begin{rcases}
y(\tau_2)=\frac{1}{L^*},&\text{for } 0\leqslant i \leqslant N_\theta\\
y(\tau)=\frac{1}{\bar{L}},&\text{for } i> N_\theta
\end{rcases}\text{for } 1\leqslant N\leqslant N_\theta
\\
\\
\begin{rcases}
y(\tau_2)=\frac{1}{L^*},&\text{for } 0\leqslant i \leqslant N\\
y(\tau)=\frac{1}{\bar{L}},&\text{for } i> N\\
\end{rcases}\text{for } N>N_\theta.
\end{cases} 
\end{equation}

Note that while the continuous and discrete scenarios proposed in this section have similarities, they have different {\color{black}structures that lead} to different properties. The primary difference between these methods is that while in continuous time the decaying term is a function of time and will vanish as $t$ grows, this is not the case in discrete scenario. The decaying term in discrete approach is a function of the sampling instant and not time. Thus, if only a few triggering instants occur, the effect of perturbation term may be considerable, regardless of the time that has passed. This important feature of discrete scenario can be seen from the examples provided in next section and shows that in contrast to continuous counterpart, the decaying term may still be kept effective for a much longer time.
\section{Illustrative Examples}\label{section examples}
In this section we illustrate the $\mathcal{L}_2$-stabilizing triggering design through several examples. Our examples are simple enough so that the $\mathcal{L}_2$-gain analysis can be done analytically, thus enabling us to provide further insight.  We show in Example \ref{exmple2} that if some of the conditions of Theorem \ref{thm main} are not satisfied, it may still be possible to relax these conditions by redefining triggering condition. In Example \ref{example3}, we replace the Euclidean vector norm with the infinity norm to obtain the $\mathcal{L}_2$-gain. This is important since this change facilitates the computation of the Lyapunov function.
\\
We continue with the following remarks, containing important points regarding the simulations.
\begin{rmk}
	The examples are constructed according to our design principle, {\it i.e.} performance is defined in ${\mathcal L}_2$-sense and the design is such that preserves the ${\mathcal L}_2$ gain of the continuous-time design. In this approach, we have purposely ignored transient behaviour and pushed the design to the extreme to save communications during transient, something that should, of course, be corrected in a more realistic design. 
	The simulations indeed show a deterioration of the transient response. This should be interpreted as indicative that, in general, ${\mathcal L}_2$ performance does not, in any way, imply good transient behaviour. 
\end{rmk} 
\begin{rmk}
	Note that the plots for verification of $\mathcal{L}_2$-gain and  system's trajectories are provided for one single initial condition. However, the discussion on number of samples and minimum inter-event times are provided based on averaging $100$ initial conditions. Thus, one should be careful that since the $\mathcal{L}_2$-gain plots depend on initial condition, no general conclusion (such as comparing the $\mathcal{L}_2$-gain of continuous-time and event-based systems) other than verification of the proposed $\mathcal{L}_2$-gain for different scenarios can be made from them.
\end{rmk}
\begin{exmp}\label{example1}
		Consider the following first order system
		\begin{equation}\label{eq_ex1}
		\dot{x}=-x^3+xw+u,~~~
		z=x,
		\end{equation}
		where $x\in \mathbb{R}$, $u=-k(x+e)$ for some $k\in\mathbb{R}_{>0}$ is the control input, $e$ is the measurement error and $w$ is the exogenous disturbance belongs to the set $\mathscr{W}_Q$ defined in (\ref{def_W_Q}). Choosing the Lyapunov function $V(x)={x^2}/{2}$ it is straightforward to show that the system is ISS with respect to $e$ and $w$. Assuming $e$ to be zero all the time, the continuous-time system is finite gain locally $\mathcal{L}_2$-stable. To show this, we take $W(x)=\lambda V(x)$ for some $\lambda\in\mathbb{R}_{>0}$. Now since $\dot{V}(x)=-kx^2 -x^4+x^2w\leqslant -kx^2 +(\gamma_1-1)x^4+w^2/({4\gamma_1})$, where $\gamma_1\in\mathbb{R}_{>0}$, we will have $\dot{W}(x)\leqslant -\lambda kx^2 +\lambda w^2/({4\gamma_1})$ for $\gamma_1\leqslant 1$. As a consequence, the minimum upper bound on the $\mathcal{L}_2$-gain of system (\ref{eq_ex1}) is $1/(2\sqrt{\gamma_1 k})$ (when $\gamma_1=1$). Finally by choosing $U=V+W$ we have $\dot{U}(x)\leqslant \lambda{w^2}/{4}-\lambda kz^2-kx^2+w^2/4+(1+\lambda)k|x||e|$, which by restricting $e$ and $w$ to satisfy $|e|\leqslant{{c} |x|}/{(1+\lambda)}$ and $|w(t)|\leqslant\gamma_3(|x(t)|)=2\sqrt{\bar{{c}}k}|x(t)|$, reads as $\dot{U}(x)\leqslant\lambda{w^2}/{4}-\lambda kz^2-(1-{c}-\bar{{c}})kx^2$. Thus we can design ${c}$ and $\bar{{c}}$ so that ${c}+\bar{{c}}<1$ and hence ensure that the event-based system is finite gain $\mathcal{L}_2$-stable. Also it is not difficult to verify that $\dot{U}(x)<0$ and hence $|x|$ monotonically converges to zero. This enables us to find $Q$ assuming $x_0\in\mathscr{X}_0$, {\it i.e.,} $|x_0|\leqslant \varepsilon$. Indeed, we can write $|w(t)|\leqslant2\sqrt{\bar{{c}}k}|x(t)|\leqslant 2\sqrt{\bar{{c}}k}|x_0|$. Hence taking $Q=2\sqrt{\bar{{c}}k}\varepsilon$ guarantees $w(t)\in\mathscr{W}_Q$.
		
We continue the discussion carried out above, numerically. Taking $k=1$, $\varepsilon= 1$, ${c}=0.5$, $\bar{{c}}=0.45$, $Q=1.34$, $\lambda=0.5$, $\kappa=15$, $\zeta=1.6$, $\hat{\kappa}=1.5$, $\theta=1$ and $\Delta=1.1$, we arrive at the execution rule $|e|={|x|}/{3}$. Consequently we have $\dot{U}(x)\leqslant|w|^2/8-|z|^2/2$, where $U(x)={3x^2}/{4}$. It then follows that the event-based system is finite gain $\mathcal{L}_2$-stable with zero bias and has $\mathcal{L}_2$-gain less than or equal to ${1}/{2}$.
To confirm the value of $\mathcal{L}_2$-gain numerically, we integrate $\dot{U}(x)-|w|^2/8+|z|^2/2\leqslant 0$ to get $U(x) - U(x_0) -{\frac{1}{8}}\int_{0}^{t}  |w(\tau)|^2d\tau+{\frac{1}{2}}\int_{0}^{t}|z(\tau)|^2d\tau\leqslant 0$
which by defining $\Gamma=\frac{1}{2}$, $\mu=2U$ and using positive definiteness of $U$ reduces to
\begin{equation}\label{eq_ex_gamma}
\frac{\int_{0}^{t}|z(\tau)|^2d\tau}{\int_{0}^{t} |w(\tau)|^2d\tau}\leqslant {\Gamma^2}+\frac{\mu(x_0)}{\int_{0}^{t} |w(\tau)|^2d\tau}
\end{equation}
and is verified in \cref{fig:gamma2} for $x_0=1$.
\begin{figure}[H]
	\vspace{-1em} 
	\centering
	\includegraphics[width=1.01\columnwidth]{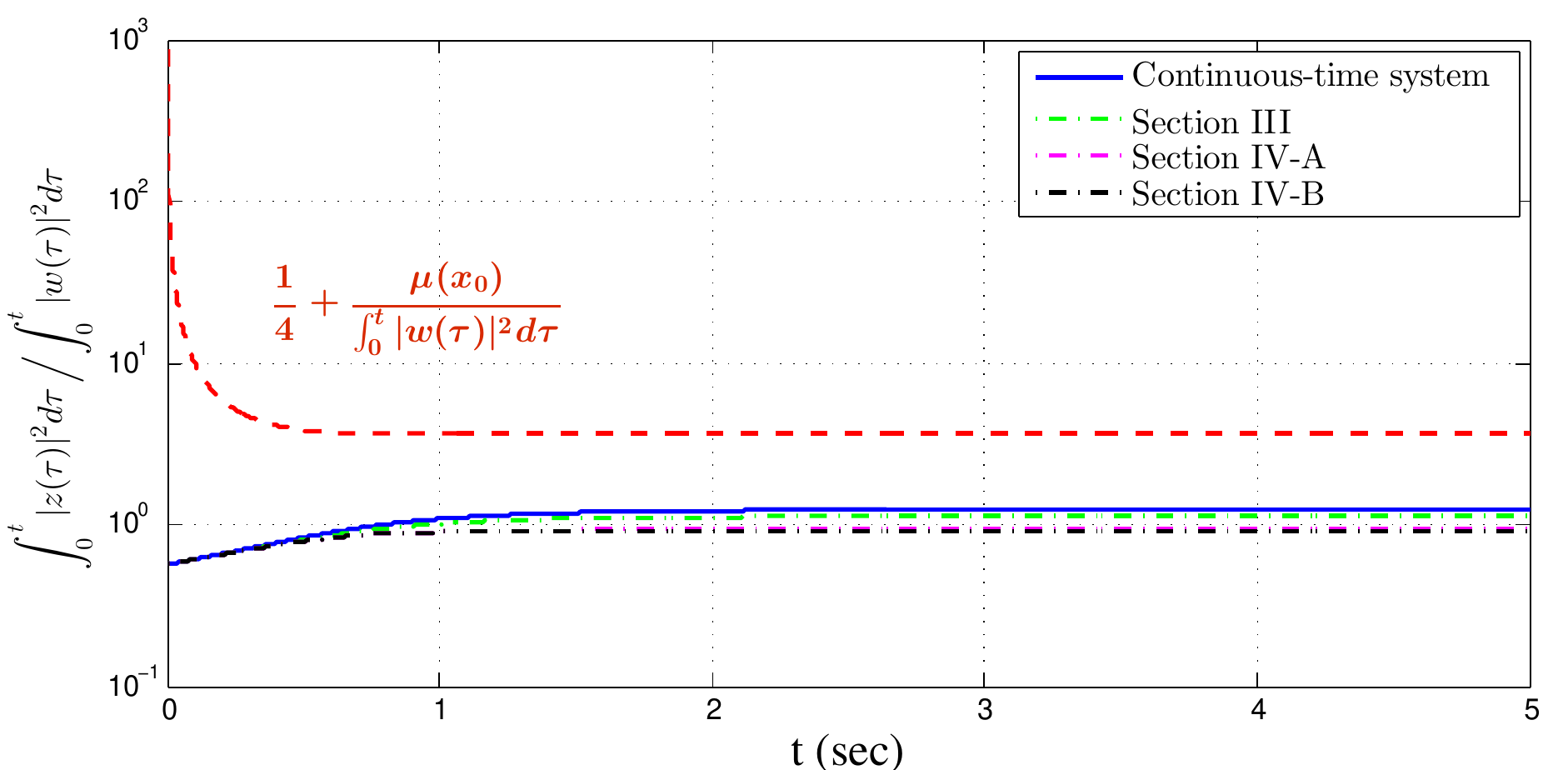}
	\vspace{-1.5em}
	\caption{Verification of $\mathcal{L}_2$-gain.}
	\label{fig:gamma2}
	\vspace{-1em} 	
\end{figure} 
Also the corresponding state trajectory of the system is shown in \cref{fig:gamma1}.
\begin{figure}[H]
	\vspace{0em} 
	\centering
	\includegraphics[width=1.01\columnwidth]{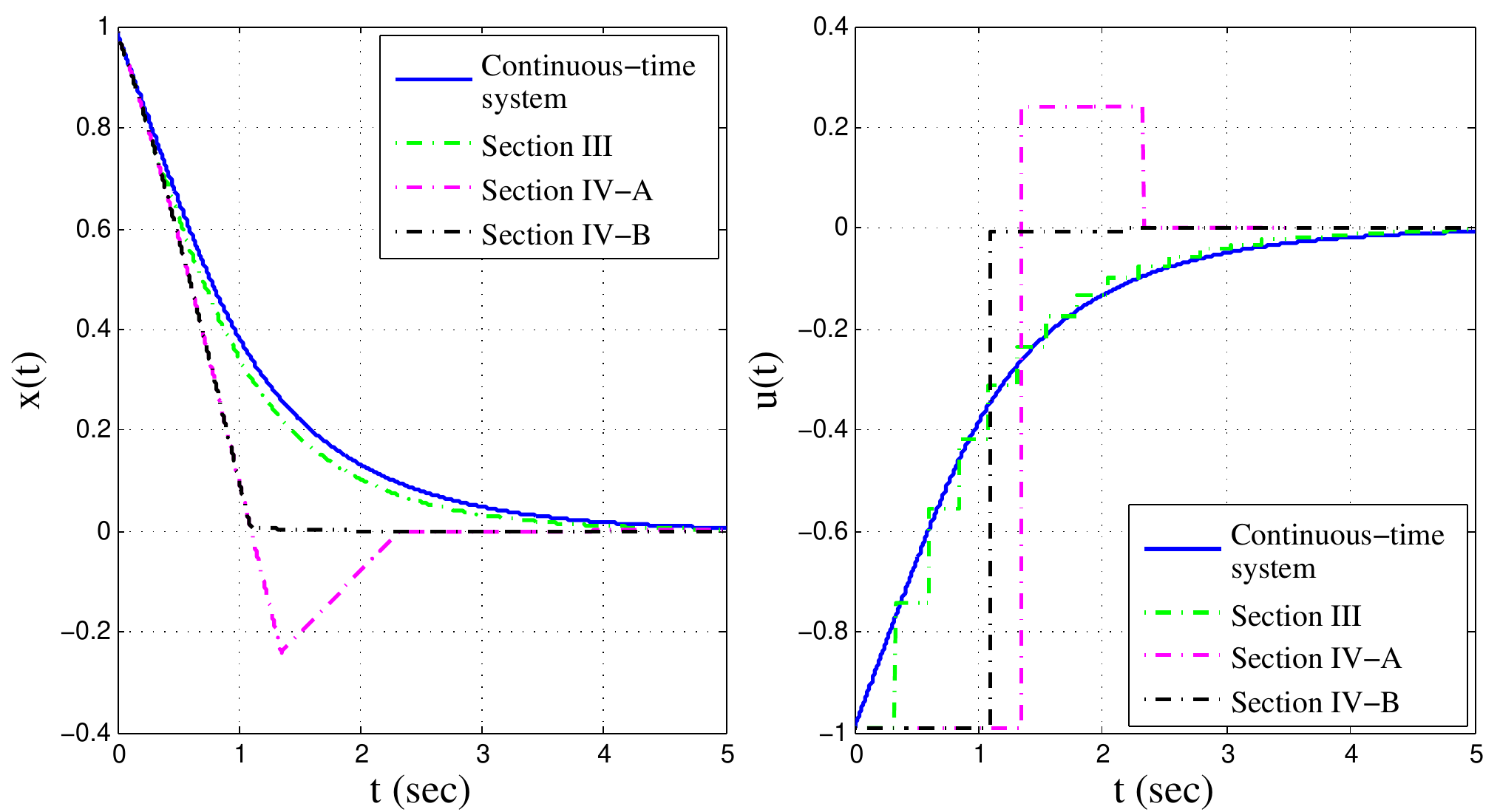}
	\vspace{-1em}
	\caption{System's trajectory (Left). Actuator signal (Right).}
	\label{fig:gamma1}
	\vspace{-1em}	
\end{figure} 
It is worth noticing that the $\mathcal{L}_2$-gain preserving nature of our proposed method can be inferred from \cref{fig:gamma2} as the curves for the event-based scenarios lie under the one for continuous-time system. 		
Also a comparison of the number of triggering instants and the minimum intersampling period is given in the following table, where we average the results obtained from $100$ initial conditions, uniformly distributed in $[-1,1]$. The results of Table \ref{tab:table1} clearly suggests that the effectiveness of the methods proposed in Section \ref{section improve inter-execution time} on the sampling rate and intersampling interval diminishes with the passing of time.
\begin{table}[H]
	\vspace{-0.5em}
	\centering
	\caption{Comparison of different scenarios.}
	\label{tab:table1}
	\begin{tabular}{ccccc}
		\toprule[1.5pt]
		&Simulation & Section \ref{section main result} & \multicolumn{2}{c}{Section \ref{section improve inter-execution time}} \\
		&time (sec) & & \ref{section_cont} & \ref{subsection discrete} \\
		\midrule
		\multirow{3}{*}{Number of samples} & 10 & 40 & 11 & 10 \\
		& 30 & 120 & 48 & 38 \\
		& 100 & 400 & 286 & 318 \\
		\cline{2-5}
		\multirow{3}{*}{Min inter-event time} & 10  & 0.24 & 0.66 & 1.1 \\
		& 30 & 0.24 & 0.46 & 0.27\\
		& 100 & 0.24 & 0.25 & 0.25 \\
		\bottomrule[1.5pt]
	\end{tabular}
\end{table}
\end{exmp}

The proof of Theorem \ref{thm main execution time} suggests that nonzero inter-event times can be guaranteed if instead of condition (ii) in Theorem \ref{thm main}, the function $\psi^{-1}(\bar{\beta}_1/c)$ is Locally Lipschitz-continuous in $\mathbb{R}^n$. Neither of the these conditions hold in the next examples, however, we can still prove this important property for the event-based system through defining a new triggering condition.\\
\begin{exmp}\label{exmple2}
In the next example, we consider the following second order system 
\begin{equation}
\begin{cases}
\dot{x}_1=x_2, \\
\dot{x}_2=-h(x_1)+u+w,\\
z=x_2.
\end{cases}
\end{equation}
where $u$ is the control input, $w$ is the exogenous disturbance and is restricted to satisfy $|w|\leq 1$ and $z$ is the measured output. We design $u=-k(x_2+e)$ where $e$ is the measurement error in $x_2$. The nonlinear function $h:\mathbb{R}\mapsto\mathbb{R}$ is assumed to be in the sector $[c_1,c_2]$, {\it{i.e.,}} $c_1r^2\leqslant rh(r)\leqslant c_2r^2$ for any $r\in \mathbb{R}$.
We first show that, in view of (\ref{last_help}), the system is ISS with respect to $e$ and $w$. To this end, let us consider Lyapunov function $V(x)=\frac{1}{2}x^TPx+2\int_{0}^{x_1}h(r)dr$,
where $x=[x_1~x_2]^T$ and $P=[1~1;1~2]$.
Then by choosing $k=1$ we have $\dot{V}(x)=-x_1h(x_1)-x_2^2+2x_2(-e+w)+x_1(-e+w)$. Next we can rewrite the last two terms in $\dot{V}$ as $2x_2(-e+w)=-\frac{1}{4}(x_2+4e)^2-\frac{1}{4}(x_2-4w)^2+\frac{1}{2}x_2^2+4e^2+4w^2$ and $x_1(-e+w)=-\frac{1}{4}(x_1+2e)^2-\frac{1}{4}(x_1-2w)^2+\frac{1}{2}x_1^2+e^2+w^2$.
Assuming $h$ to be in the sector $[1,2]$, we conclude that $r^2\leqslant rh(r)\leqslant 2 r^2$. Taking this into account, we obtain $\dot{V}(x)\leqslant -\bar{\sigma}(|x|)+\beta_1(|e|)+\beta_2(|w|)$, where $\bar{\sigma}(r)=r^2/2$ and $\beta_1(r)=\beta_2(r)=5r^2$.

We also claim that when $e\equiv 0$, the continuous-time system is finite gain locally $\mathcal{L}_2$-stable. To see this, consider the Lyapunov function $W(x)=\lambda V(x)$ for some $\lambda\in\mathbb{R}_{>0}$. Then since $\dot{V}(x)=-x_1h(x_1)-x_2^2+2x_2w+x_1w = -x_1^2 (1-\epsilon_1)-x_2^2(1-\epsilon_2)+(\frac{1}{4}\epsilon_1^{-1}+\epsilon_2^{-1})w^2 = -\epsilon_1 (x_1-\frac{1}{2}\epsilon_2^{-1}w)^2-\epsilon_2(x_2-\epsilon_2^{-1}w)^2$,
we conclude $\dot{W}(x)\leqslant\lambda(1-\epsilon_2)z^2+\lambda(\epsilon_1^{-1}/4+\epsilon_2^{-1})w^2$, {\it i.e.,} the continuous-time system has local $\mathcal{L}_2$-gain less than or equal to $\sqrt{(4\epsilon_1+\epsilon_2)/(4\epsilon_1\epsilon_2(1-\epsilon_2))}$. The minimum value of this upper bound on the $\mathcal{L}_2$-gain of the system is $4.4861$ and obtained by setting $\epsilon_1=1$ and $\epsilon_2=0.4721$.

To verify condition (\ref{bound_on_w}) we restrict $w$ to satisfy $|w(t)|\leqslant\gamma_3(|x(t)|)$, where $\gamma_3(r)= \sqrt{{\bar{{c}}}/{2}}r$ for some $\bar{{c}}\in(0,1)$ is a solution to the inequality (\ref{cond_id}). So far we have showed that Assumptions \ref{assumption2}, \ref{assumption1} hold. Therefore it suffices to verify conditions (i)-(iii) in Theorem \ref{thm main} hold as well. Condition (iii) is readily hold for functions $f$ and $k$. Also, condition (i) holds for $\sigma_3(r)=\lambda(\|P\|+2c_2)r$
since we have
\begin{equation*}
|{\frac{\partial W}{\partial x}}(x)|=\lambda\begin{vmatrix}\begin{bmatrix}
x_1+x_2+2h(x_1) \\ x_1+2x_2
\end{bmatrix}\end{vmatrix}\leqslant \lambda(\|P\|+2c_2)|x|.
\end{equation*}
Condition (ii) in Theorem \ref{thm main} is not satisfied for the given functions $\bar{\sigma}$, $\beta_1$. However, we will redefine functions $\psi$ and $\bar{\beta}_1$ in (\ref{execution rule}) and show the results of theorem are still valid. To this end, let us start with (\ref{last_help}) which can be written as $\nabla V(x) \cdot f(x,k(x+e),w)\leqslant -(1-c_0)\bar{\sigma}(|x|)+\beta_1(|e|)$ for some $c_0\in(0,1)$ when $|w|\leqslant \gamma_3(|x|)$. This is true since choosing $c_0\geqslant 5\bar{c}$ ensures $\beta_2(|w|)\leqslant c_0\bar{\sigma}(|x|)$. Therefore, (\ref{dot V}) reduces to $\dot{U}(x)\leqslant-(1-c_0)\bar{\sigma}(|x|)+\beta_1(|e|)+\sigma_0(|x|)\beta_0(|e|)+\Gamma^2|w|^2-|z|^2$. Using the definition of $\beta_1$ in this example, we can write $\beta_1(|e|)+\sigma_0(|x|)\beta_0(|e|)=\beta_1(|e|)+L_fL_k|e|\sigma_3(|x|)=(\sqrt{5}|e|+{L_fL_k}\sigma_3(|x|)/({2\sqrt{5}}))^2-{L_f^2L_k^2}\sigma_3^2(|x|)/{20}$ and hence
\begin{eqnarray}
\dot{U}(x)&{}\leqslant{}& -(1-c_0)\bar{\sigma}(|x|)-\frac{L_f^2L_k^2}{20}\sigma_3^2(|x|)+(\sqrt{5}|e|\nonumber\\&&{+}\: \frac{L_fL_k}{2\sqrt{5}}\sigma_3(|x|))^2+ {\Gamma}^2|w|^2-|z|^2.
\end{eqnarray} 
Therefore, we can define $\psi(r)\doteq \sqrt{\bar{\sigma}(r)+{L_f^2L_k^2}{\sigma_3}^2(r)/{20}}-{L_fL_k}\sigma_3(r)/({2\sqrt{5}})$ and $\bar{\beta}_1(r)\doteq \sqrt{5}r$.
Thus if for some $c\in(0,1-c_0)$ the next triggering of control task occurs when 
\begin{eqnarray}
\sqrt{5}|e|\geqslant\sqrt{{c}\bar{\sigma}(|x|)+\frac{L_fL_k^2}{20}\sigma_3^2(|x|)}-\frac{L_fL_k}{2\sqrt{5}}\sigma_3(|x|)
\end{eqnarray}
we conclude that $\dot{U}(x)\leqslant -(1-c_0-c)\bar{\sigma}(|x|)+ {\Gamma}^2|w|^2-|z|^2$. As a consequence, the event-based system has the local $\mathcal{L}_2$-gain less than or equal to $4.4861$.
Also one can check the local Lipschitz-continuity of $\psi^{-1}(\bar{\beta}_1/c)$ in $\mathbb{R}^n$ which is necessary to prove Zeno-freeness property for the system. 

To find $Q$, we have to find functions $\sigma_1$ and $\sigma_2$ so that (\ref{ISS V eq}) holds. Since $x_1^2\leqslant 2\int_{0}^{x_1}h(r)dr\leqslant 2x_1^2$ and $V(x)=(x^TPx)/2+2\int_{0}^{x_1}h(r)dr$, one can choose $\sigma_1(r)={\Sigma_{\text{min}}(P_1)}r^2/2$ and $\sigma_2(r)={\Sigma_{\text{max}}(P_2)}r^2/2$, where $\Sigma_{\text{max}}(A)$ (respectively $\Sigma_{\text{min}}(A)$) denotes maximum (respectively minimum) eigenvalue of matrix $A$, and $P_1=[3~1;1~2]$, $P_2=[5~1;1~2]$.
Thus we can take $Q=L_{\gamma_3}\bar{\varepsilon}=L_{\gamma_3}\sigma_1^{-1}(\sigma_2(\varepsilon))=\varepsilon\sqrt{(\bar{{c}}\Sigma_{\text{max}}(P_2))/(2\Sigma_{\text{min}}(P_1))}$, where $\varepsilon$ is the upper bound on the norm of admissible initial conditions.  
For numerical simulations we take $\varepsilon=1$, $\lambda=10^{-3}$, ${c}=0.7$, $\bar{{c}}=0.05$, $\kappa=\hat{\kappa}=50$, $\zeta=\theta=1$, $\Delta=4$ and $Q=0.62$. 
The verification of $\mathcal{L}_2$-gain of the system for $x_0=[0.87~0.5]^T$ is presented in Fig. \ref{fig:ex2-1} as it suggests (\ref{eq_ex_gamma}) holds for $\Gamma=4.4861$.
\begin{figure}[H]
	\vspace{-0.6em} 
	\centering
	\includegraphics[width=1.0\columnwidth]{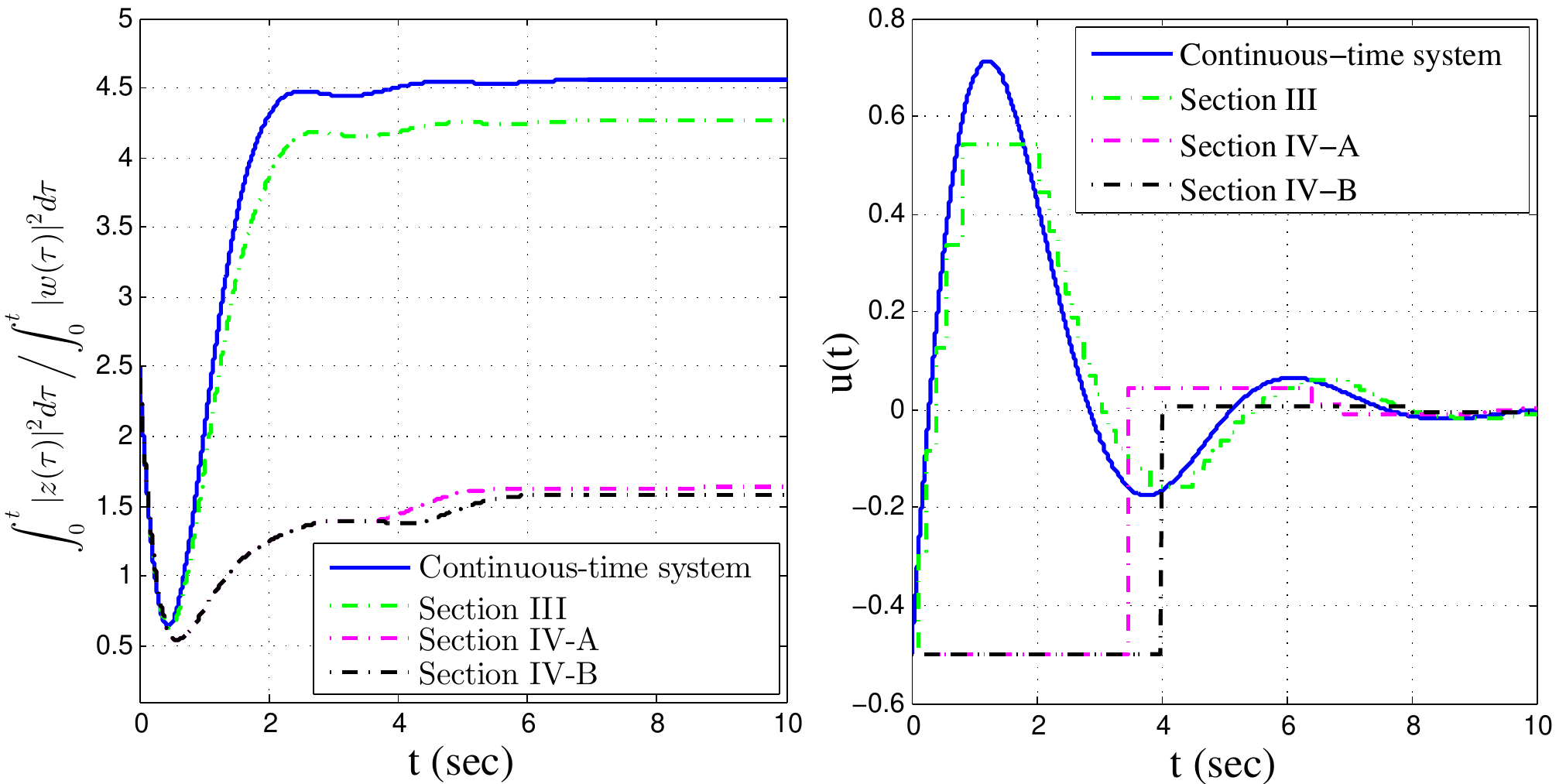}
	\vspace{-1.5em}
	\caption{Verification of $\mathcal{L}_2$-gain (Left). Actuator signal (Right).}
	\label{fig:ex2-1}
	\vspace{-1em}	
\end{figure}
The state trajectories of the system is also plotted in Fig. \ref{fig:ex2-2}.
\begin{figure}[H]
	\vspace{-1em} 
	\hspace{-0.5em}
	\includegraphics[width=1.01\columnwidth]{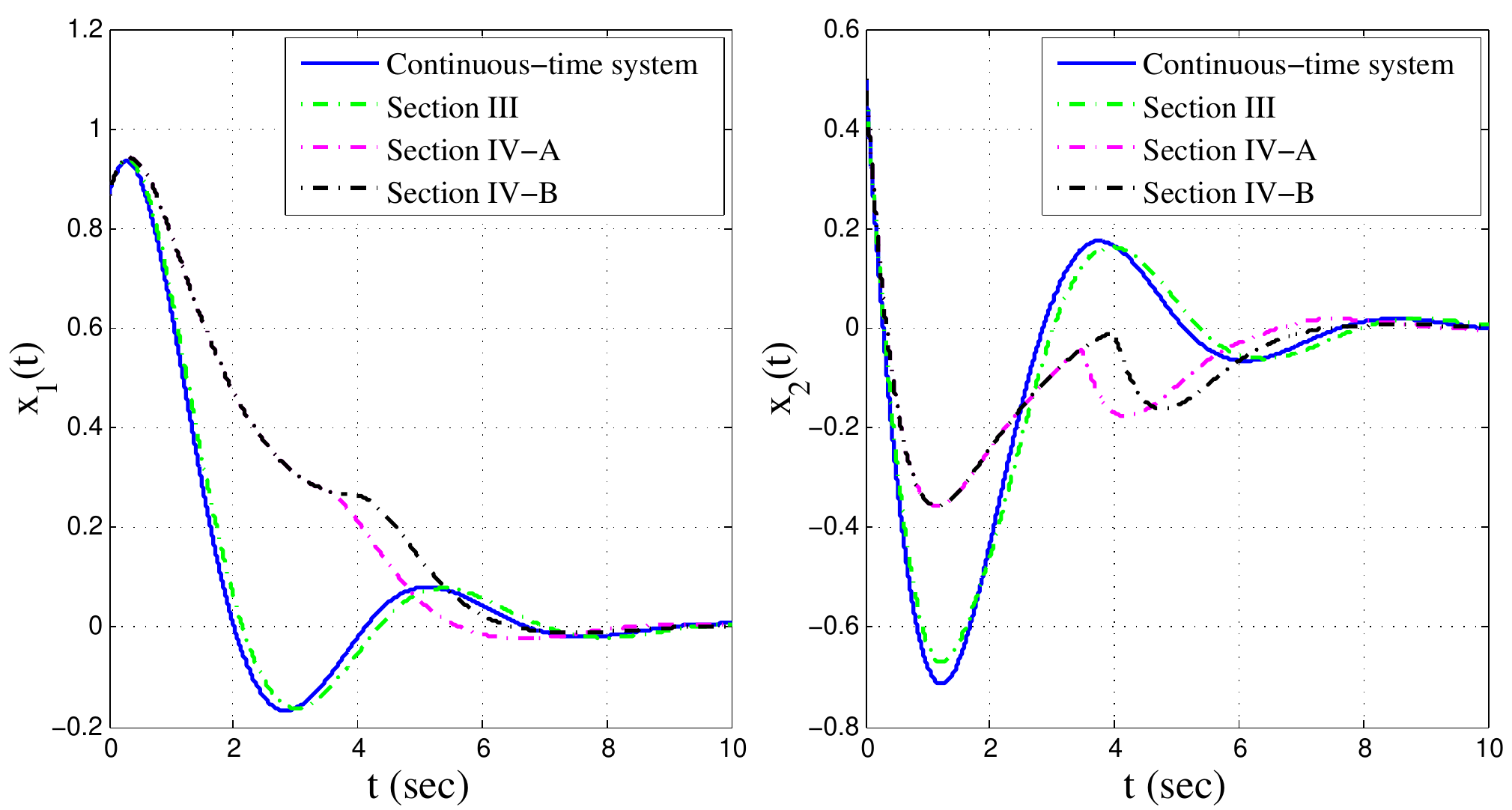}
	\vspace{-1.2em}
	\caption{System's trajectories.}
	\label{fig:ex2-2}
	\vspace{-1em}	
\end{figure}
Finally, a comparison of the number of triggering instants and the minimum intersampling period is given in Table 	\ref{tab:table2}. To this end, we consider $100$ initial conditions uniformly distributed in circle of radius $1$ and average the obtained results.
\begin{table}[H]
	\vspace{-0.5em}
	\centering
	\caption{Comparison of different scenarios.}
	\label{tab:table2}
	\begin{tabular}{ccccc}
		\toprule[1.5pt]
		&Simulation & Section \ref{section main result} & \multicolumn{2}{c}{Section \ref{section improve inter-execution time}} \\
		&time (sec) & & \ref{section_cont} & \ref{subsection discrete} \\
		\midrule
		\multirow{3}{*}{Number of samples} & 10& 47 & 6 & 3 \\
		& 30 & 139 & 89 & 8 \\
		& 100 & 466 & 415 & 70 \\
		\cline{2-5}
		\multirow{3}{*}{Min inter-event time} & 10& 0.09 & 1.21 & 4 \\
		& 30 & 0.09 & 0.1 & 4\\
		& 100 & 0.09 & 0.09 & 0.49 \\
		\bottomrule[1.5pt]
	\end{tabular}
\end{table}
\end{exmp}

In the next example, we apply the refsults of Theorem \ref{thm main} but replacing the Euclidean vector norm with the infinity norm.\\   
\begin{exmp}\label{example3}
Using similar notations as in Example \ref{exmple2}, we define the following second order system
\begin{equation}\label{ex2}
\begin{cases}
\dot{x}_1=x_2-bx_1,\\
\dot{x}_2=-ax_1^3+u+w,\\
z=x_2,\\
\end{cases} 
\end{equation}
where $|w|\leqslant 1$.
Defining Lyapunov function $V(x)=ax_1^4/4+|x|^2/2$, where $x=[x_1~x_2]^T$, we will have $\dot{V}(x)=x_1x_2-bx_1^2-abx_1^4+x_2u+x_2w$, which by taking $u=-(x_1+e_1)-(x_2+e_2)$ can be written as
\setlength{\arraycolsep}{0.0em}
\begin{eqnarray}\label{V1_dot}
\dot{V}(x) \leqslant -bx_1^2-abx_1^4-x_2^2-\sqrt{2}{|x|}_{\infty}|e|+{|x|}_{\infty}|w|
\end{eqnarray}
where $e_1$ and $e_2$ are the measurement errors in $x_1$ and $x_2$, respectively and $e=[e_1~e_2]^T$.
Then in view of the following inequality
\setlength{\arraycolsep}{0.0em}
\begin{eqnarray}
bx_1^2+abx_1^4+\frac{1}{4}x_2^2\geqslant 
\begin{cases}
b{|x|}_{\infty}^2+ab{|x|}_{\infty}^4, & \text{if } |x_1| > |x_2|, \\
\frac{1}{4}{|x|}_{\infty}^2, & \text{otherwise},\\
\end{cases} \nonumber
\end{eqnarray} 
we conclude that $bx_1^2+abx_1^4+x_2^2/4\geqslant \bar{\sigma}({|x|}_{\infty})$, where function $\bar{\sigma}(r)=\smash{\displaystyle\min}\{br^2+abr^4,r^2/4\}$ is of class $\mathcal{K}_{\infty}$. This enables us to write (\ref{V1_dot}) as
\begin{equation}
\dot{V}(x)\leqslant -\bar{\sigma}({|x|}_{\infty})+\sqrt{2}{|x|}_{\infty}|e|+{|x|}_{\infty}|w|.
\end{equation}
To show finite gain stability of continuous-time system, consider $W(x)=\lambda V(x)$ as the Lyapunov function. Thus for $e\equiv0$, we have $\dot{V}(x)=-bx_1^2-abx_1^4-x_2^2+x_2w\leqslant-z^2+zw$ which by using $zw={\hat{\epsilon}}z^2/{2}+w^2/({2\hat{\epsilon}})-{\hat{\epsilon}}{(z- w/{\hat{\epsilon}})}^2/{2}$ for some $\hat{\epsilon}\in\mathbb{R}_{>0}$, gives $\dot{W}(x)\leqslant-\lambda(1-{\hat{\epsilon}}/{2})z^2+\lambda w^2/({2\hat{\epsilon}})$. As a consequence, it is not difficult to show that the minimum upper bound on the $\mathcal{L}_2$-gain of continuous-time system (\ref{ex2}) can be achieved by choosing $\hat{\epsilon}=1$ and is equal to $1$. This value, however, may be improved by a different choice of Lyapunov function $W(x)$. 

Defining $\sigma_3(r)\doteq  \lambda r+a\lambda r^3$ for $r\in\mathbb{R}_{\geq0}$, we have ${|\frac{\partial W}{\partial x}(x)|}_{\infty}\leqslant \sigma_3({|x|}_{\infty})$ since 
\begin{eqnarray}
{|\frac{\partial W}{\partial x}(x)|}_{\infty}&{}={}&\lambda
{\begin{vmatrix}
\begin{bmatrix}
ax_1^3+x_1 \\ x_2
\end{bmatrix}
\end{vmatrix}}_{\infty}\nonumber \\
&{}\leqslant{}& 
\begin{cases}
a\lambda{|x|}_{\infty}^3+\lambda{|x|}_{\infty}, & \text{if } |x_1|(1+ax_1^2) > |x_2|, \\
\lambda{|x|}_{\infty}, & \text{otherwise},\\
\end{cases} \nonumber
\end{eqnarray}
and hence ${|\frac{\partial W}{\partial x}(x)|}_{\infty}\leqslant \smash{\displaystyle\max}\{\lambda{|x|}_{\infty},\lambda{|x|}_{\infty}+a\lambda{|x|}_{\infty}^3\}=\lambda{|x|}_{\infty}+a\lambda{|x|}_{\infty}^3$.
Therefore, by taking $U=V+W$ and following the similar lines as in deriving (\ref{dot V 2}),  we can write
\begin{eqnarray}
\dot{U}(x)&{}\leqslant{}& -\frac{\lambda}{2}z^2{+}\frac{\lambda}{2}w^2{+}{|x|}_{\infty}\Big{(}{-|x|}_{\infty}\smash{\displaystyle\min}\{b{+}ab{|x|}_{\infty}^2,\frac{1}{4}\}\nonumber\\&{}{}&{+}\sqrt{2}|e|{+}|w|{+}2\lambda L_fL_k(1+a{|x|}_{\infty}^2)|e|
\Big{)},\nonumber
\end{eqnarray}
where we used the fact that 
\begin{eqnarray}
&{}{}&\frac{\partial W}{\partial x}(x)(f(x,k(x+e),w)-f(x,k(x),w))\leqslant \nonumber\\ &{}{}&2{|\frac{\partial W}{\partial x}(x)|}_{\infty}{|f(x,k(x+e),w)-f(x,k(x),w)|}_\infty\leqslant \nonumber \\ &{}{}&2{|\frac{\partial W}{\partial x}(x)|}_{\infty}{|f(x,k(x+e),w)-f(x,k(x),w)|}\leqslant \nonumber \\ &{}{}&2L_fL_k|e|{|\frac{\partial W}{\partial x}(x)|}_{\infty}\nonumber.
\end{eqnarray}
We note that since $u=k(x)=-(x_1+x_2)$, it can be easily inferred that $L_k=\sqrt{2}$.
Now assuming $|w(t)|\leqslant\gamma_3({|x(t)|}_{\infty})$ where $\gamma_3(r)= \smash{\displaystyle\min}\{br+abr^3,r/4\}\bar{{c}}$ and taking ${c}\in(0,1-\bar{{c}})$, we conclude that if the execution of control task occurs when
\begin{equation}\label{ex3_tc}
|e|\geqslant {c} \frac{{|x|}_{\infty} \smash{\displaystyle\min}\{b+ab{|x|}_{\infty}^2,\frac{1}{4}\}}{\sqrt{2}(1+2\lambda L_f(1+a{|x|}_{\infty}^2))},
\end{equation}
the system (\ref{ex2}) is finite gain local $\mathcal{L}_2$-stable with zero bias and local $\mathcal{L}_2$-gain $\leqslant 1$. To find $Q$, let $\sigma_1(x)\leqslant V(x) \leqslant \sigma_2(x)$, where $\sigma_1(r)=r^2/2$ and $\sigma_2(r)=(2+a)r^2/4$. As a consequence, assuming initial conditions to be norm bounded by $\varepsilon$, we can take $Q=L_{\gamma_3}\bar{\varepsilon}=L_{\gamma_3}\sigma_1^{-1}(\sigma_2(\varepsilon))=L_{\gamma_3}\varepsilon\sqrt{1+{a}/{2}}$, which by choosing $b>{1}/{4}$, reduces to $Q=({\bar{{c}}\varepsilon}\sqrt{1+{a}/{2}})/{4}$.
In simulations, let $a=1$, $b=10$, $\epsilon=1$, ${c}=0.5$, $\bar{{c}}=0.45$, $\kappa=10$, $\bar{\kappa}=10$, $\zeta=1$, $\theta=5$, $\Delta=1$ and $Q=0.138$. Therefore, the only parameter left to study
the system's response is $\lambda$ which appears in triggering condition (\ref{ex3_tc}). We start our simulation with $\lambda=1$, however, the effect of this parameter on our results will be discussed later. Similar to the past examples, in the next table, we give a comparison of number of samplings and minimum intersampling times over different scenarios. The results are, indeed, the average over $100$ initial conditions uniformly distributed in the circle of radius $1$.      
\begin{table}[H]
	\vspace{-0.5em}
	\centering
	\caption{Comparison of different scenarios.}
	\label{tab:table3}
		\vspace{-0.5em}
	\begin{tabular}{ccccc}
		\toprule[1.5pt]
		&Simulation & Section \ref{section main result} & \multicolumn{2}{c}{Section \ref{section improve inter-execution time}} \\
		&time (sec) & & \ref{section_cont} & \ref{subsection discrete} \\
		\midrule
		\multirow{3}{*}{Number of samples} & 10 & 420 & 8 & 10 \\
		&30 & 1190 & 23 & 30 \\
		&100 & 3882 & 75 & 100 \\
		\cline{2-5}
		\multirow{3}{*}{Min inter-event time} &10 & 0.007 & 0.61 & 1 \\
		&30  & 0.007 & 0.57 & 1\\
		&100 & 0.007 & 0.53 & 1 \\
		\bottomrule[1.5pt]
	\end{tabular}
		\vspace{-1.2em}
\end{table}
Recalling from Definition \ref{def gain}, the system (\ref{ex2}) has local $\mathcal{L}_2$-gain $\leqslant {\Gamma}$ if for any $T$ we have (\ref{eq_ex_gamma}).
The local $\mathcal{L}_2$-gain of the system is then verified in Fig. \ref{fig:ex3-4} for $x_0=[0.87~0.5]^T$ and $\Gamma=1$. Also the corresponding state trajectories is presented in Fig. \ref{fig:ex3-3}.
\begin{figure}[H]
	\vspace{-0.6em} 
	\hspace{-0.4em}
	\includegraphics[width=1.01\columnwidth]{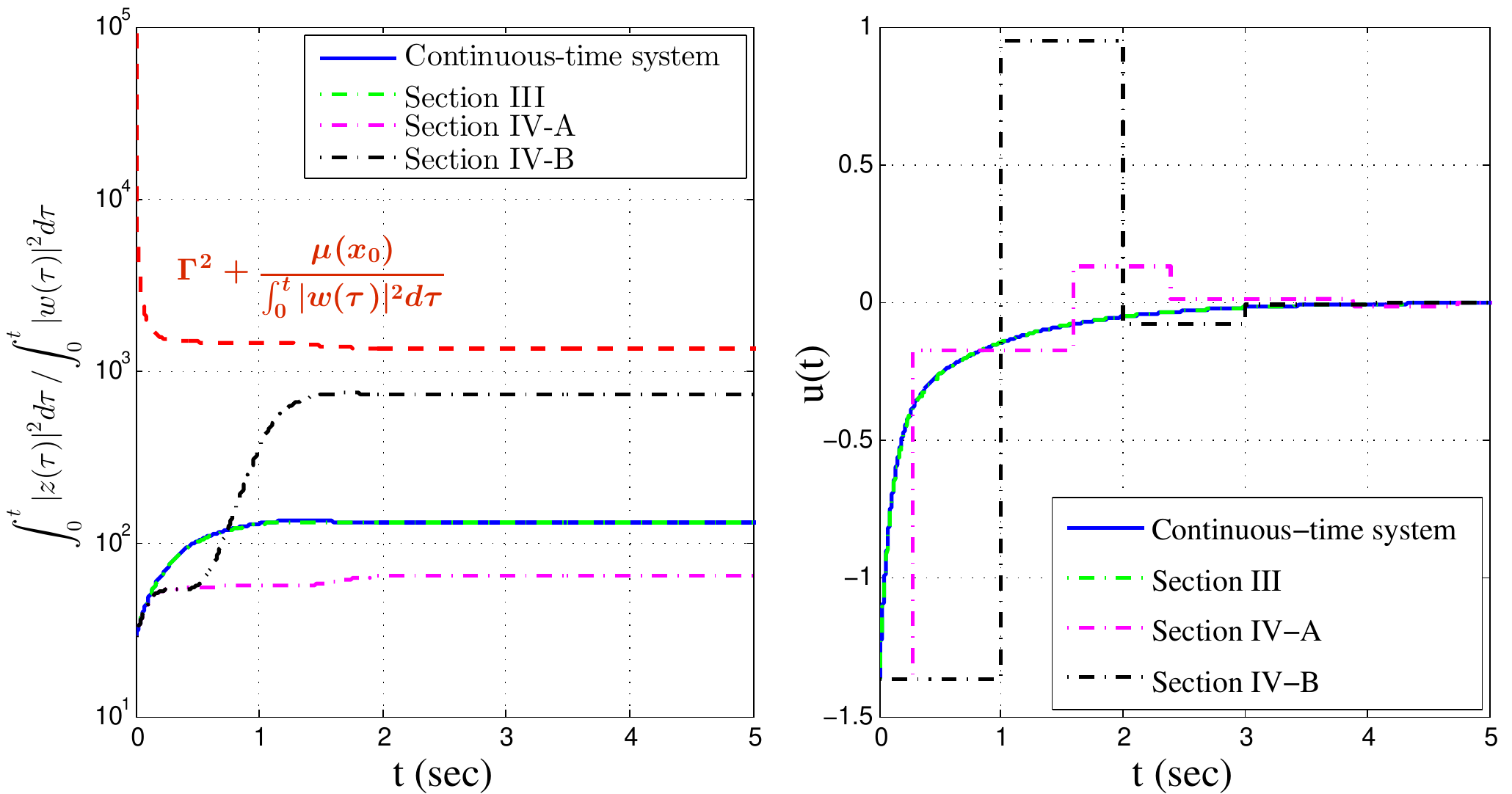}
	\vspace{-1.3em}
	\caption{Verification of $\mathcal{L}_2$-gain (Left). Actuator signal (Right).}
	\label{fig:ex3-4}
	\vspace{-1.3em}	
\end{figure}
\begin{figure}[H]
	\vspace{-0.5em} 
	\hspace{-0.3em}
	\includegraphics[width=1.01\columnwidth]{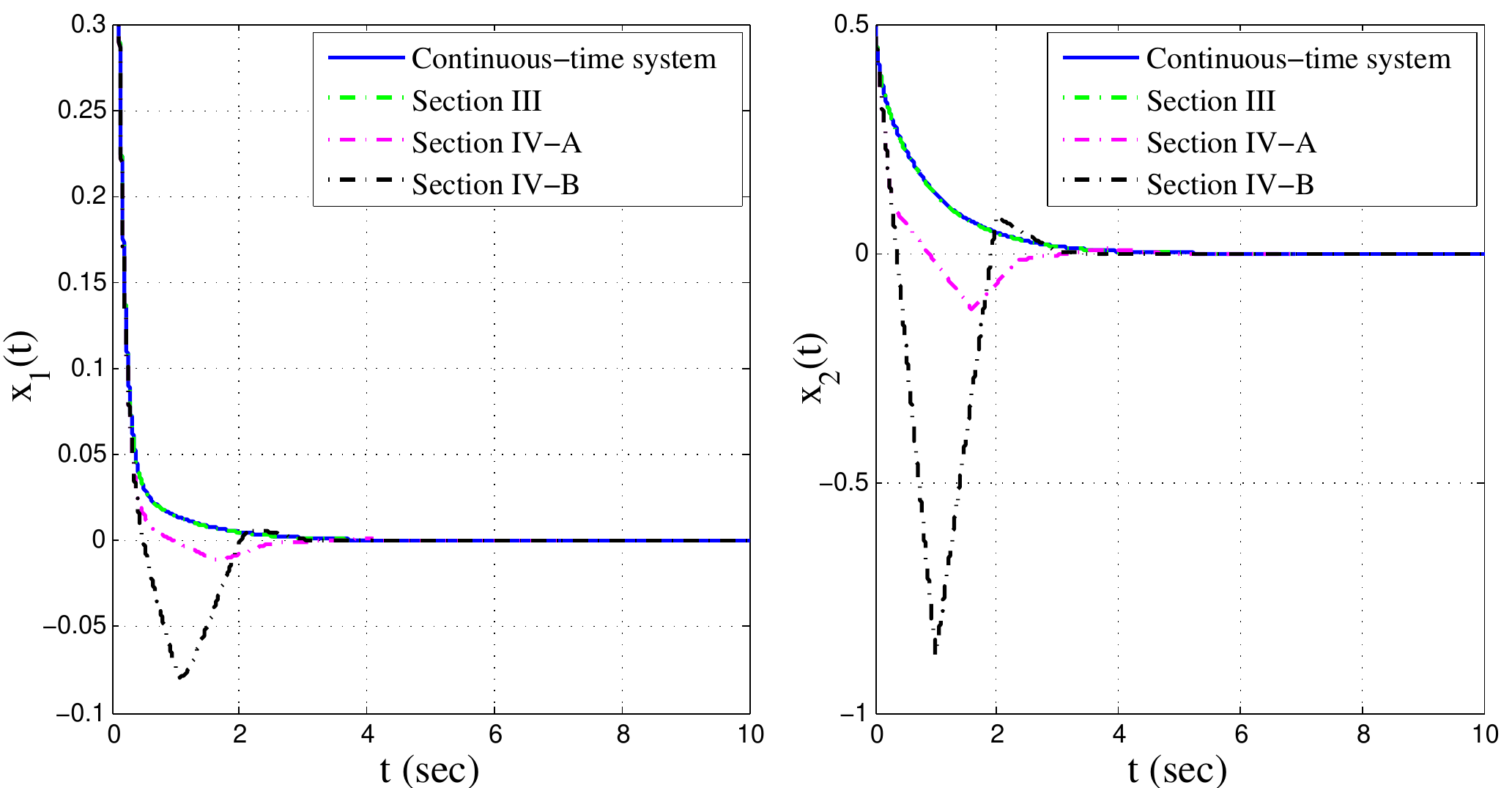}
	\vspace{-1.2em}
	\caption{System's trajectories.}
	\label{fig:ex3-3}
	\vspace{-1em}	
\end{figure} 
Finally, the effect of parameter $\lambda$ on the above results in the first $100$ seconds of response is investigated in Table \ref{tab:table4}. It suggests that $\lambda$ has negligible effect on the triggering numbers and minimum inter-event times using the methods of Sections \ref{section_cont}, \ref{subsection discrete}. However, choosing $\lambda>10^{-2}$ degrades the efficiency of the results of Section \ref{section main result}, significantly. 
\begin{table}[H]
	\vspace{-0em}	
	\centering
	\caption{Investigating the effect of parameter $\lambda$.}
	\label{tab:table4}
	\begin{tabular}{cccccc}
		\toprule[1.5pt]
		&&\multicolumn{4}{c}{$\lambda$}\\
		\cline{3-6}
		&Method of& & & & \\
		& Section &  $10^{-3}$ & $10^{-2}$ & $10^{-1}$ & 1\\
		\midrule
		\multirow{3}{*}{Number of samples} &  \ref{section main result} &149 & 152 & 176 & 420 \\
		& \ref{section_cont}& 8 & 8 & 8 & 8\\
		& \ref{subsection discrete}& 10 & 10 & 10 & 10 \\
		\cline{2-6}
		\multirow{3}{*}{Min inter-event time}&  \ref{section main result} &  0.023 & 0.022 & 0.019 & 0.007 \\
		& \ref{section_cont} &  0.670 & 0.669 & 0.680 & 0.610 \\
		& \ref{subsection discrete} & 0.999 & 0.999 & 0.999 & 0.999 \\
		\bottomrule[1.5pt]
	\end{tabular}
\end{table}
\end{exmp}

\begin{exmp}
	The following example illustrates the necessity of using a local $\mathcal{L}_2$ theory. The example shows that while under	arbitrary perturbations $w$ in $\mathcal{L}_2$ space, the event times are not necessarily guaranteed to be isolated, the local notion serves to exclude Zeno phenomenon. Consider the following linear example from \cite{event-separation}
	\begin{eqnarray}\label{eq_exmp_last}
	\dot{x}(t)=Ax(t)+Bu(t)+w(t),~~u(t)=Kx(t),
	\end{eqnarray}
	where $A$, $B$, $K$ are matrices of appropriate dimensions and the controller is applied in an event-based fashion. The desired output is taken as $z(t)=x(t)$. Assume $t_0=0$, $x_0\neq 0$, and the triggering condition $|e(t)|\geqslant p|x(t)|$ for some $p\in\mathbb{R}_{>0}$, it is shown in \cite{event-separation} that under the following choice of disturbance
	\begin{eqnarray}\label{exmp_dist}
	w(t)=((t-1)A+(t_i-1)BK)x_0-x_0,~t\in[t_i,t_{i+1})
	\end{eqnarray}
	for $t\in[0,1]$ and zero elsewhere (which is a signal in $\mathcal{L}_2(\mathbb{R}_{\geq0})$ space), the state and triggering instants are analytically given by $x(t)=(1-t)x_0$ and $t_i=1-(1+p)^{-i}$, $i\in\mathbb{Z}_{\geq0}$, respectively. It is then obvious that event times has an accumulation point at $t=1$. To address this issue, \cite{event-separation} suggests using the input-to-state practically stable (ISpS) property instead of ISS condition (\ref{ISS eq}). The proposed method, however, is not applicable to the problem studied in this paper since the $\mathcal{L}_2$-gain performance of the event-based system can not be guaranteed.   
		
	Note that the above discussion suggests that when $w$ is an arbitrary signal in $\mathcal{L}_2(\mathbb{R}_{\geq0})$, as in (\ref{exmp_dist}), the execution rules of the form (\ref{execution rule}) does not exclude the Zeno-behaviour. However, in this paper our solution to this problem is to restrict $w$ to be in the admissible space $\mathscr{W}_{Q}$ and also satisfy condition (\ref{bound_on_w}) with $\gamma_3(r)\doteq \hat{c} r$, $\hat{c}\in\mathbb{R}_{>0}$. The price we paid is then the local character of the results. We remark that $w$ defined in (\ref{exmp_dist}) does not satisfy (\ref{bound_on_w}), and hence is not a counter example of the local thoery. This is because (\ref{bound_on_w}) is violated near $t=1$. 
	
	Indeed, applying the results of Theorem \ref{thm main execution time} one can show that limiting $w$ as above, the triggering instants are separated at least by 
	\begin{eqnarray*}
	\tau=\frac{1}{\|A\|+\hat{c}}\ln{(1+\frac{\|A\|+\hat{c}}{p(\|A\|+\|BK\|+\hat{c})+\|BK\|})}.
	\end{eqnarray*}
\end{exmp}	

\section{Conclusion}
This paper addresses the disturbance rejection problem of nonlinear event-based systems. Assuming the existence of a pre-designed control law with desirable local $\mathcal{L}_2$ performance characteristics, we propose a triggering condition that preserves finite gain local $\mathcal{L}_2$-stability of the original continuous-time design. Our formulation is rather general; {\it i.e.} we consider a nonlinear plant and assume that disturbances are bounded by a Lipschitz-continuous function of the state. We also show that, in the absence of external disturbances, the control law render the origin asymptotically stable.

In addition to stability and disturbance rejection, we also study the intersampling behaviour of the proposed event-triggering condition. 
First we show that the inter-event time period is lower bounded by a nonzero constant and focus on enlarging this constant. We show
that, regardless of the construction of the event-triggered mechanism, the inter-event time period increase is actually lower bounded by a constant.
Increasing the value of this constant can be done at the expense of relaxing the stability properties of the design. 

\section{Appendix}
\begin{newproofof}
	\textit{Theorem \ref{thm_equiv}.}
	We need to show that if there exist class $\mathcal{K}_\infty$ functions $\sigma$, $\gamma_i$ ($i=1,2$) so that $\nabla V(\xi)\cdot f(\xi,k(\xi+\mu),w)\leqslant-\sigma(|\xi|)$ holds for any $\xi \in \mathbb{R}^n$, any $\mu \in \mathbb{R}^n$ and any $w\in \mathscr{W}_Q$ such that $|\xi|\geqslant \gamma_1(|\mu|)+\gamma_2(|w|)$ then one can find class $\mathcal{K}_\infty$ functions $\bar{\sigma}$, $\beta_i$ ($i=1,2$) such that $\nabla V(\xi)\cdot f(\xi,k(\xi+\mu),w)\leqslant-\bar{\sigma}(|\xi|)+\beta_1(|\mu|)+\beta_2(|w|)$ and vice versa. 
	Let us start by assuming $\nabla V(\xi)\cdot f(\xi,k(\xi+\mu),w)\leqslant-\sigma(|\xi|)$ for $|\xi|\geqslant \gamma_1(|\mu|)+\gamma_2(|w|)$. Then we can say that $\nabla V(\xi)\cdot f(\xi,k(\xi+\mu),w)+\sigma(|\xi|)\leqslant \bar{\beta}(|\mu|,|w|)$ where 
	\setlength{\arraycolsep}{0.0em}
	\begin{eqnarray}
		\bar{\beta}(|\mu|,|w|)=\smash{\displaystyle{\max}}\{\nabla V(\xi)\cdot&{}{}&  f(\xi,k(\xi+r),s)+\sigma(|\xi|)|~|r|\leqslant|\mu|,\nonumber \\ &{}{}& |s|\leqslant |w|, |\xi|\leqslant \gamma_1(|r|)+\gamma_2(|s|) \}.\nonumber
	\end{eqnarray}
	Defining class $\mathcal{K}_\infty$ functions $\beta_1(|\mu|)\doteq \bar{\beta}(|\mu|,|\mu|)$ and $\beta_2(|w|)\doteq \bar{\beta}(|w|,|w|)$ it is not difficult to verify that $\bar{\beta}(|\mu|,|w|)\leqslant\beta_1(|\mu|)$ for $|\mu|\geq |w|$ and $\bar{\beta}(|\mu|,|w|)\leqslant\beta_2(|w|)$ otherwise. Therefore we conclude that $\bar{\beta}(|\mu|,|w|)\leqslant\beta_1(|\mu|)+\beta_2(|w|)$ that proves one part of the claim. To prove the other side, assume that we have $\nabla V(\xi)\cdot f(\xi,k(\xi+\mu),w)\leqslant-\bar{\sigma}(|\xi|)+\beta_1(|\mu|)+\beta_2(|w|)$. Then we can write $\nabla V(\xi)\cdot f(\xi,k(\xi+\mu),w)\leqslant-\bar{\sigma}(|\xi|)/2$ for $\bar{\sigma}(|\xi|)/4\geqslant \beta_1(|\mu|)$ and $\bar{\sigma}(|\xi|)/4\geqslant \beta_2(|w|)$. Finally defining $\gamma_i\doteq \bar{\sigma}^{-1}(4\beta_i)$ ($i=1,2$), we may conclude that $\nabla V(\xi)\cdot f(\xi,k(\xi+\mu),w)\leqslant-\bar{\sigma}(|\xi|)/2$ for $|\xi|\geqslant \gamma_1(|\mu|)+\gamma_2(|w|)$ which completes the proof. 
\end{newproofof}

\begin{newproofof}
\textit{Lemma \ref{lem unify}.}
(a) This is an immediate consequence of Theorem \ref{thm_equiv}.
(b) We need to show that under conditions I-III, there exists a class $\mathcal{K}_\infty$ function $\gamma$ so that $\nabla V(\xi)\cdot f(\xi,k(\xi+\mu),w)\leqslant-\sigma(|\xi|)$ for $|\xi|\geqslant \gamma(|\mu|)$. To this end, let us start with conditions II and III that suggest $\gamma_4\circ(|\xi|-\gamma_2(|w|))\geqslant \gamma_4\circ(\gamma_{id}-\gamma_2\circ\gamma_3)(|\xi|)\geqslant |\xi|$. Now taking $\gamma=\gamma_4\circ\gamma_1$ we can say that if $|\xi|\geqslant \gamma(|\mu|)$, we have $|\xi|-\gamma_2(|w|)\geqslant \gamma_1(|\mu|)$ which, in view of condition I, implies that $\nabla V(\xi)\cdot f(\xi,k(\xi+\mu),w)\leqslant-\sigma(|\xi|)$.
\end{newproofof}

\begin{newproofof}
\textit{Lemma \ref{key_lem}.}
(if) From (\ref{lem_eq}) we may conclude that $\nabla V(\xi)\cdot f(\xi,k(\xi+\mu),w) \leqslant -\hat{\sigma}(|\xi|)$ for $c\psi(|\xi|)\geqslant \bar{\beta}_1(|\mu|)$. Then taking $\sigma=\hat{\sigma}$, $\gamma=\psi^{-1}(\bar{\beta}_1/c)$ and applying Lemma \ref{lem unify} part (a), the desired result is obtained. (only if) Starting from Assumption \ref{assumption1}, by adding and subtracting $\sigma_0(|\xi|)\beta_0(\mu)$ term to the right hand side of inequality (\ref{assump_eq_V_dot}), we may write $\nabla V(\xi)\cdot f(\xi,k(\xi+\mu),w) \leqslant -(1-c)\bar{\sigma}(|\xi|)-\sigma_0(|\xi|)\beta_0(|\mu|)$ for $\beta_1(\mu)+\sigma_0(|\xi|)\beta_0(|\mu|)\leqslant c\bar{\sigma}(|\xi|)$. Now defining functions $\psi(r)\doteq {\bar{\sigma}(r)}/({1+\sigma_0(r)})$, $\bar{\beta}_1(r)\doteq \smash{\displaystyle\max}{\{\beta_1(r),\beta_0(r)\}}$, we claim that if $c\psi(|\xi|)\geqslant \bar{\beta}_1(|\mu|)$ we have $\beta_1(\mu)+\sigma_0(|\xi|)\beta_0(|\mu|)\leqslant c\bar{\sigma}(|\xi|)$. This is true since $c\bar{\sigma}(|\xi|)\geqslant (1+\sigma_0(|\xi|))\cdot \smash{\displaystyle\max}{\{\beta_1(\mu),\beta_0(\mu)\}}\geqslant \beta_1(\mu)+ \sigma_0(|\xi|)\beta_0(|\mu|)$. Therefore, if $c\psi(|\xi|)\geqslant \bar{\beta}_1(|\mu|)$, (\ref{lem_eq}) holds for $\hat{\sigma}=(1-c)\bar{\sigma}$ and hence the proof is complete.
\end{newproofof}

\begin{newproofof}
\textit{Theorem \ref{thm main}.}
	Let us start with Assumption \ref{assumption1} which, in view of proof of Lemma \ref{key_lem}, implies the existence of ${\bf{C}}^1$ function $V$ such that
	\begin{equation}\label{ISS V1}
	\nabla V(x) \cdot f(x,k(x+e),w)\leqslant -(1-c)\bar{\sigma}(|x|)- \sigma_0(|x|)\beta_0(|e|)
	\end{equation}
	for any $x \in \mathbb{R}^n$, any $e \in \mathbb{R}^n$ and any $w\in \mathscr{W}_Q$ such that $c\psi (|x|)\geqslant \bar{\beta}_1(|e|)$.
	Now consider positive definite ${\bf{C}}^1$ function $U=V+W$, where $W$ is a positive definite ${\bf{C}}^1$ function that, in view of Assumption \ref{assumption2}, guarantees the finite gain local $\mathcal{L}_{2}$-stability of continuous-time system $\mathscr{G}_c$. We can easily write 
	\setlength{\arraycolsep}{0.0em}
	\begin{eqnarray}\label{dot V}
	\dot{U}(x)&{}={}&\nabla V(x)\cdot f(x,k(x+e),w)+ \nabla W(x)\cdot f(x,k(x),w)\nonumber\\&&{+}\: \nabla W(x)\cdot \big{(}f(x,k(x+e),w){-}f(x,k(x),w)\big{)}.
	\end{eqnarray}
	Also applying condition (i) and inequality (\ref{rmk lip eq4}) gives $\nabla W(x)\cdot \big{(}f(x,k(x+e),w)-f(x,k(x),w)\big{)}\leqslant \sigma_0(|x|)\beta_0(|e|)$.
	As a consequence, in view of (\ref{gain}), (\ref{ISS V1}) and (\ref{dot V}) we can write
	\setlength{\arraycolsep}{0.0em}
	\begin{equation}\label{dot V 2}
	\dot{U}(x)\leqslant -(1-c)\bar{\sigma}(|x|)+ {\Gamma}^2|w|^2-|h(x,w)|^2
	\end{equation}
	for any $x \in \mathbb{R}^n$, any $e \in \mathbb{R}^n$ and any $w\in \mathscr{W}_Q$ such that $c\psi (|x|)\geqslant \bar{\beta}_1(|e|)$. 
	Thus under event condition (\ref{execution rule}) we obtain $H_{\Gamma}(U,k(x+e))\leqslant 0$, {\it i.e.}, the event-based system $\mathscr{G}_e$ has the disturbance attenuation local $\mathcal{L}_2$-gain ${\|\mathscr{G}_e\|}_{\mathcal{L}_2}\leqslant {\Gamma}$. 
\end{newproofof}

\begin{newproofof}
\textit{Lemma \ref{pro IC}.}
	We deduce from inequality (\ref{comparison}) that $\bar{\sigma}(|x|)-\beta_1(|e|)
	-\beta_0(|e|)\sigma_0(|x|) \geqslant (1-{c})\bar{\sigma}(|x|)$ and hence $
	\bar{\sigma}(|x|)-\beta_1(|e|)\geqslant 0$. Thus we conclude from (\ref{assump_eq_V_dot}) that $\dot{V}(x)\leqslant0$ and consequently $V(x(t))\leqslant V(x(0))$ for all $t\in\mathbb{R}_{\geq0}$. Since $V$ is a radially unbounded positive definite function, we conclude that there exists $\sigma_1$, $\sigma_2\in\mathcal{K}_\infty$ so that (\ref{ISS V eq}) holds
	and hence $\sigma_1(|x(t)|)\leqslant V(x(t))\leqslant V(x(0)) \leqslant \sigma_2(|x(0)|)$. Then we can write $|x(t)|\leqslant\sigma_1^{-1}(\sigma_2(x_0))$ and since $x_0\in\mathscr{X}_0$ and $\sigma_1^{-1}$, $\sigma_2$ are class $k_{\infty}$ functions, the desired result is obtained. 
\end{newproofof}

\begin{newproofof}
	\textit{Theorem \ref{thm main execution time}.}
	From Lemma \ref{pro IC}, we have $x(t)\in\mathscr{X}$ for all $t\in\mathbb{R}_{\geq 0}$. Now in view of Properties \ref{pro lip}, \ref{pro beta} it can be inferred that function ${\psi}^{-1}({\bar{\beta}_1}/{{c}})$ is Lipschitz-continuous in any compact set in $\mathbb{R}_{\geq0}$. Let us denote by $\bar{L}$ the Lipschitz constant of this function on set $\mathscr{D}_e$ defined as $\mathscr{D}_e= \{\bar{\beta}_1^{-1}( {c}\psi(s))| s\in [0,\bar{\varepsilon}]\}=[0,\bar{\beta}_1^{-1}( {c}\psi(\bar{\varepsilon}))]$. Thus we have ${\psi}^{-1}( \bar{\beta}_1(|e|)/{{c}})\leqslant \bar{L}|e|$ which suggests that a more conservative lower bound on inter-event times can be achieved when instead of (\ref{execution rule}), the next triggering of control task occurs when  $\bar{L}|e|\geqslant|x|$.
	Following the same procedure as in (\cite{tabuada}, Theorem $\textrm{III}.1$), we can upperbound the dynamics of $y\doteq {|e|}/{|x|}$ as $\dot{y}\leqslant\big{(}1+y\big{)}{|\dot{x}|}/{|x|}$, which using (\ref{dotx}) reads as
	\begin{equation}\label{error upper bound}
	\begin{gathered}
	\dot{y}\leqslant
	L_f \Big{(}1+y\Big{)}\Big{(}L_k+L_{\gamma_3}+1+L_ky\Big{)}.
	\end{gathered}
	\end{equation}
	Thus the inter-execution times are lower bounded by the solution $\tau$ of $y(\tau)={1}/{\bar{L}}$, where $y$ is the solution to 
	\begin{equation}\label{differnetial eq}
	\dot{y}=L_f(1+y)(L+L_k y),~y(0)=0
	\end{equation}
	with $L=L_k+L_{\gamma_3}+1$. It then follows that the lower bound on inter-event times is
	\begin{equation}\label{lower bound}
	0<\tau=\frac{1}{L_f(L-L_k)}\ln{(1+\frac{L-L_k}{L\bar{L}+L_k})}.
	\end{equation} 
\end{newproofof}

\begin{newproofof}
\textit{Property \ref{pro lip}.}
	Let us define $\varepsilon_m\doteq \max_{r\in\mathscr{D}_x}\{r\}$. Also let $L_{\sigma_0}$ and $L_{\bar{\sigma}^{-1}}$ be the Lipschitz constants of functions $\sigma_0$ and $\bar{\sigma}^{-1}$ on compact sets $\{\psi^{-1}(r)|r\in\mathscr{D}_x\}=[0,\psi^{-1}(\varepsilon_m)]$ and $\{\bar{\sigma}(\psi^{-1}(r))| r\in\mathscr{D}_x\}=[0,\bar{\sigma}(\psi^{-1}(\varepsilon_m))]$, respectively. Using the fact that $\bar{\sigma}$ and $\sigma_0$ are class $\mathcal{K}_\infty$ functions, one can write
	\setlength{\arraycolsep}{0.0em}
	\begin{eqnarray}
	|\psi(r)-\psi(\tilde{r})| &{}={}& \Big{|}\frac{\bar{\sigma}(r)}{1+\sigma_0(r)}-\frac{\bar{\sigma}(\tilde{r})}{1+\sigma_0(\tilde{r})}\Big{|}
	\nonumber\\ &{}\geqslant{}& 				\Big{|}\frac{\big{(}1+\sigma_0(r)\big{)}\Delta_{r,\tilde{r}}({\bar{\sigma}})- \bar{\sigma}(r)\Delta_{r,\tilde{r}}({\sigma_0})}{(1+\sigma_0(\varepsilon_m))^2}\Big{|}
	\nonumber\\ &{}\geqslant{}& \frac{\big{(}1+\sigma_0(r)\big{)}\big{|}\Delta_{r,\tilde{r}}({\bar{\sigma}})\big{|}-
		\bar{\sigma}(r)\big{|}\Delta_{r,\tilde{r}}({\sigma_0})\big{|}}{(1+\sigma_0(\varepsilon_m))^2}\nonumber
	\end{eqnarray}
	for any $r,\tilde{r}\in\{s|\psi(s)\in \mathscr{D}_x\}$, where functional $\Delta_{r,\tilde{r}}$ is defined as $\Delta_{r,\tilde{r}}({\varphi})\doteq\varphi(r)-\varphi(\tilde{r})$ for some function $\varphi$. The Lipschitz-continuity of functions $\sigma_0$ and $\bar{\sigma}^{-1}$ imply that $|\Delta_{r,\tilde{r}}({\sigma_0})|\leqslant L_{\sigma_0}|r-\tilde{r}|$ and $|r-\tilde{r}|\leqslant L_{\bar{\sigma}^{-1}}|\Delta_{r,\tilde{r}}({\bar{\sigma}})|$ which together with Lemma \ref{pro IC} reduces the above inequality to
	\setlength{\arraycolsep}{0.0em}
	\begin{eqnarray}
	|\psi(r)-\psi(\tilde{r})| &{}\geqslant{}& \frac{{L_{\bar{\sigma}^{-1}}^{-1}}-L_{\sigma_0}\bar{\sigma}(\varepsilon_m)}{(1+\sigma_0(\varepsilon_m))^2}|r-\tilde{r}|. 
	\end{eqnarray}
\end{newproofof}

\begin{newproofof}
	\textit{Corollary \ref{cor w-Iss}.}
	From Remark \ref{rmk w=0} we conclude that $\beta_1(|e|)\leqslant{c}\bar{\sigma}(|x|)$ between triggering instants. Then assuming $w=0$ and taking (\ref{assump_eq_V_dot}) into account, we can write $	\nabla V (x)\cdot f(x,k(x+e),0)\leqslant-(1-{c})\bar{\sigma}(|x|) < 0$, {\it i.e.,} $x=0$ is an asymptotically stable point for disturbance-free system $\mathscr{G}_e$. The above argument is global since from Assumption \ref{assumption1}, $V$ is radially unbounded. 
\end{newproofof}

\begin{newproofof}
	\textit{Theorem \ref{thm practical0}.}
	Following the similar lines as in the proof of Theorem \ref{thm main}, we can upper bound $\dot{U}$ as
	\begin{equation}\label{new dotV}
	\dot{U}(x)\leqslant -(1-{c})\bar{\sigma}(|x|)+{\kappa} e^{-\zeta t}+ {\Gamma}^2|w|^2-|h(x,w)|^2
	\end{equation}
	for any $x \in \mathbb{R}^n$, any $e \in \mathbb{R}^n$ and any $w\in \mathscr{W}_Q$ such that $c\tilde{\psi} (|x|)\geqslant \bar{\beta}_1(|e|)$. 
	Integrating (\ref{new dotV}) from $0$ to $T\in\mathbb{R}_{>0}$
	and using the positive definiteness of $U$ we obtain $\int_{0}^{T}{|h(x(t),w(t))|^2}dt\leqslant {{\Gamma}^2}\int_{0}^{T}{|w(t)|^2}ds+ {\kappa}(1-e^{-\zeta T})/{\zeta}+U(x_0)$, which by applying Definition \ref{def gain} with $\eta={\kappa}/{\zeta}$ and $\mu=U$, completes the proof. 
\end{newproofof}

\begin{newproofof}
	\textit{Corollary \ref{cor w-Iss2}.}
	It can be inferred from execution rule (\ref{new execution}) that between successive triggering instants we have $\beta_1(|e|)\leqslant \bar{\beta}_1(|e|)\leqslant{c}\bar{\sigma}(|x|)+{\kappa}e^{-\zeta t}$. Then from (\ref{assump_eq_V_dot}), we can upper bound $\dot{V}$ as
	\begin{equation}\label{convergence eq}
	\dot{V}(x)\leqslant-(1-{c})\bar{\sigma}(|x|)+\kappa e^{-\zeta t}.
	\end{equation}
	Defining $\bar{{c}}=1-{c}$, we conclude that $\dot{V}(x)<0$ for $|x|> \bar{\sigma}^{-1}({\kappa e^{-\zeta t}}/\bar{{c}})$. Now define compact set $\Lambda_i=\{x\in \mathbb{R}^n||x|\leqslant \bar{\sigma}^{-1}({\kappa e^{-\zeta i}}/\bar{{c}})\}$ for $i\in\mathbb{Z}_{\geq 0}$. Also the set of boundary points of $\Lambda_i$ is defined as $\partial \Lambda_i=\{x\in \mathbb{R}^n||x|= \bar{\sigma}^{-1}({\kappa e^{-\zeta i}}/\bar{c})\}$ for $i\in\mathbb{Z}_{\geq 0}$. We denote by $m_i$ the argument of maximum value of $V(x)$ over the set $\partial \Lambda_i$, {\it i.e.}, $m_i=\smash{\displaystyle\operatorname{arg\,max}_{x\in\partial \Lambda_i}}{V(x)}$. Next define compact set $\Omega_i\doteq \{x\in\mathbb{R}^n|V(x)\leqslant V(m_i)\}$ for $i\in\mathbb{Z}_{\geq 0}$. Clearly $\Omega_i$ is positive invariant under the dynamics of event-based system $\mathscr{G}_e$ for $t\geqslant i$. We claim that $\Omega_i$ is the global attracting set of system $\mathscr{G}_e$ for $t\geqslant i$. To see this, let us define the complement of $\Omega_i$ in $\mathbb{R}^n$ as $\Omega_i^c=\{x\in \mathbb{R}^n|V(x)>V(m_i)\}$. If $x\in \Omega_i^c$, we conclude that $|x|>m_i$ and since $m_i\in \partial \Lambda_i$ we deduce that $|x|>\bar{\sigma}^{-1}({\kappa e^{-\zeta i}}/\bar{{c}})$.  
	Then since $t\geqslant i$ it follows that $|x|>\bar{\sigma}^{-1}({\kappa e^{-\zeta t}}/\bar{{c}})$ and consequently $\dot{V}(x)<0$ for all $x\in \Omega_i^c$ which confirms our claim. For $t\geqslant i+1$, however, $\Omega_{i+1}\subset \Omega_i$ is the new global attracting set of the event-based system $\mathscr{G}_e$. Thus the sequence of positive invariant attracting sets ${\{\Omega_i\}}_{i\in\mathbb{Z}_{\geq0}}$ with $\Omega_0\supset \Omega_1\supset \cdots \supset \Omega_i\supset \cdots$ shrinks to the origin as $i\rightarrow\infty$ (since $m_i$ converges to $0$) which confirms the convergence of trajectories of system $\mathscr{G}_e$ to the origin. 
\end{newproofof}

\begin{newproofof}
	\textit{Theorem \ref{thm practical}.}
	First we apply (\ref{convergence eq}) to conclude that $\dot{V}(x)\leqslant-{\bar{{c}}}\bar{\sigma}(|x|)+\kappa$ and hence $\dot{V}(x)<0$ for $|x|>\bar{\sigma}^{-1}({\kappa}/{\bar{{c}}})$, where $\bar{c}=1-c$. Then we just need to show conditions (a)-(c) in Definition \ref{practical stability} hold. To satisfy condition (a) we can choose $\delta(\epsilon)$ such that $0<\delta(\epsilon)<\epsilon$. For condition (b) one can choose $\upsilon(r)=\bar{\sigma}^{-1}({\kappa}/{\bar{{c}}})$ for $r\leqslant\bar{\sigma}^{-1}({\kappa}/{\bar{{c}}})$
	and some $\upsilon(r)>r$ for $r>\bar{\sigma}^{-1}({\kappa}/{\bar{{c}}})$. Finally, to satisfy condition (c) we consider two cases. For $r\leqslant\bar{\sigma}^{-1}({\kappa}/{\bar{{c}}})$ we can choose $T(r,\epsilon)$ to be any positive number since trajectories of the system $\mathscr{G}_e$ do not leave the ball $\{x\in\mathbb{R}^n| |x|\leqslant \bar{\sigma}^{-1}({\kappa}/{\bar{{c}}})\}$ and hence $|x(t)|<\epsilon$ for all $t\geqslant0$ and all $\epsilon>\bar{\sigma}^{-1}({\kappa}/{\bar{{c}}})$. However, for $r>\bar{\sigma}^{-1}({\kappa}/{\bar{{c}}})$ we need a more detailed argument. Let us choose $T^\prime$ such that $|x(T^\prime)|=\epsilon$ and
	integrate (\ref{convergence eq}) from $0$ to $T^\prime$ to obtain $V(\epsilon)-V(x_0) \leqslant -\bar{{c}}\int_{0}^{T^\prime}\bar{\sigma}(|x(t)|)dt+\kappa\int_{0}^{T^\prime}e^{-\zeta t}dt \leqslant -\bar{{c}}T^\prime\bar{\sigma}(|\epsilon|)+\frac{\kappa}{\zeta}(1-e^{-\zeta{T^\prime}})$.
	Since $V(\epsilon)-V(r)\leqslant V(\epsilon)-V(x_0)$ we can upperbound $T^\prime$ as the solution to inequality $\bar{{c}}T^\prime\bar{\sigma}(|\epsilon|)+\kappa e^{-\zeta{T^\prime}}\leqslant V(r)-V(\epsilon)+\kappa$. One can find a more conservative upper bound on $T^\prime$ by neglecting the exponential term in left hand side, {\it i.e.}, $\bar{{c}}T^\prime\bar{\sigma}(|\epsilon|)\leqslant V(r)-V(\epsilon)+\kappa-\kappa e^{-\zeta{T^\prime}}\leqslant V(r)-V(\epsilon)+\kappa$ and obtain $T^\prime\leqslant (V(r)-V(\epsilon)+\kappa)/{(\bar{{c}}\bar{\sigma}(|\epsilon|))}$. This is exactly what if we integrate the more conservative inequality $\dot{V}(x)\leqslant-\bar{{c}}\bar{\sigma}(|x|)+\kappa$ instead. Choosing $T(r,\epsilon)>(V(r)-V(\epsilon)+\kappa)/(\bar{{c}}\bar{\sigma}(|\epsilon|))$ completes the proof. 
\end{newproofof}

\begin{newproofof}
	\textit{Theorem \ref{thm practical2}.}
	It can be inferred from (\ref{2 time}) that $t_i\leqslant t_i^\prime$ and hence we have 
	$\bar{\beta}_1(|e(t)|) \leqslant c \check{\psi}(|x(t)|)$
	for $t\in[t_{i-1},t_i)$. Then following similar lines as the proof of Theorem \ref{thm main}, for $t\in[t_{i-1},t_i)$ we obtain $\dot{U}(x)\leqslant-(1-{c})\bar{\sigma}(|x|)+\frac{}{}{\hat{\kappa}\theta^i}/{i!}+ {\Gamma}^2|w|^2-|h(x,w)|^2$ for any $x \in \mathbb{R}^n$, any $e \in \mathbb{R}^n$ and any $w\in \mathscr{W}_Q$ such that $c\check{\psi} (|x|)\geqslant \bar{\beta}_1(|e|)$. Integrating this inequality from $0$ to some $T{\geqslant0}$, we arrive at
	\setlength{\arraycolsep}{0.0em}
	\begin{eqnarray}
	&{}{}&U(x(T))\leqslant U(x_0)+\int_{0}^{T}({\Gamma}^2|w(t)|^2-|h(x(t),w(t))|^2)dt\nonumber\\ &{}{}&+
	{\hat{\kappa}}\Big{\{}
	\int_{t_0=0}^{t_1}{\frac{\theta^1}{1!}}dt+\cdots+ \int_{t_{i-1}}^{t_{i}}{\frac{\theta^{i}}{i!}}dt+\cdots+\int_{t_{N-1}}^{T}{\frac{\theta^{N}}{{N}!}}dt\Big{\}},\nonumber
	\end{eqnarray}
	where we assume $N$ triggering instants (including the first one at $t_0=0$) occur until $t=T$, {\it i.e.}, $t_{N-1}=\smash{\displaystyle\operatorname{max}_{t_i\leqslant T}\{t_i\}}$. Now since $U(x(T))\geqslant 0$ we conclude that $\int_{0}^{T}{|h(x(t),w(t)|^2}dt \leqslant {U(x_0)}+{{\Gamma}^2}\int_{0}^{T}{|w(t)|^2}dt + {\hat{\kappa}}\smash{\displaystyle\max_{1\leq i\leq N}}\{t_i-t_{i-1}\}\sum_{i=1}^{N}{{\theta^i}/{i!}}$ and hence
	\begin{eqnarray}\label{dissipative discrete}
	\int_{0}^{T}{|h(x(t),w(t)|^2}dt \leqslant {U(x_0)} +{{\Gamma}^2}\int_{0}^{T}{|w(t)|^2}dt+{\hat{\kappa} \Delta}e^\theta.~~~~~
	\end{eqnarray}
	We then choose $\eta={\hat{\kappa}\Delta }e^\theta$ and $\mu=U$ in Definition \ref{def gain} to obtain the desired result. 
\end{newproofof}

\begin{newproofof}
	\textit{Corollary \ref{cor w-Iss discrete}.}
	Following similar lines as the proof of Corollary \ref{cor w-Iss2} we deduce that $	\beta_1(|e|)\leqslant {c}\bar{\sigma}(|x|)+{\hat{\kappa}}{\theta^i}/{i!}$
	for $t\in[t_{i-1},t_i)$ and hence from (\ref{assump_eq_V_dot}) it follows that
	\begin{equation}\label{convergence eq discrete}
	\dot{V}\leqslant-\bar{{c}}\bar{\sigma}(|x|)+{\hat{\kappa}}\frac{\theta^i}{i!}
	\end{equation}
	for $t\in[t_{i-1},t_i)$. As a consequence we conclude that $\dot{V}(x)<0$ for $|x|\geqslant \bar{\sigma}^{-1}({\hat{\kappa} \theta^i}/{(\bar{{c}} i!)})$ and $t\in [t_{i-1},t_i)$. Now define compact set $\Lambda_i=\{x\in \mathbb{R}^n||x|\leqslant \bar{\sigma}^{-1}({\hat{\kappa} \theta^i}/{(\bar{{c}}i!)})\}$ for $i\in \mathbb{Z}_{> 0}$. Then the set of boundary points of $\Lambda_i$ can be defined as $\partial\Lambda_i=\{x\in\mathbb{R}^n||x|=\bar{\sigma}^{-1}({\hat{\kappa} \theta^i}/{(\bar{{c}}i!)})\}$. We denote the argument of maximum value of $V(x)$ over set $\partial \Lambda_i$ by $m_i=\smash{\displaystyle\operatorname{arg\,max}_{x\in\partial \Lambda_i}}{V(x)}$. We remark that the discrete function ${\theta^i}/{i !}$ takes its maximum value at $i=\lfloor \theta \rfloor$ and is strictly decreasing over $i\geqslant \lfloor \theta \rfloor$. Now define compact set $\Omega_i=\{x\in\mathbb{R}^n|V(x)\leqslant V(m_i)\}$ for $i\in\mathbb{Z}_{> 0}$. Following similar lines as the proof of Corollary \ref{cor w-Iss2}, we can show that for $i\geqslant \lfloor \theta \rfloor$, $\Omega_i$ is positive invariant under dynamics of the event-based system $\mathscr{G}_e$ and moreover, is the global attracting set of this system for $t\geqslant t_i$. Since $t_i-t_{i-1}\leqslant \Delta$, we conclude that $i\rightarrow \infty$ as $t\rightarrow \infty$, {\it i.e.}, the triggering instants never terminate. Thus the sequence of positive invariant attracting sets ${\{\Lambda_i\}}_{\lfloor \theta \rfloor\leqslant i\in \mathbb{Z}_{> 0}}$ with $\Lambda_{\lfloor \theta \rfloor}\supset \Lambda_{\lfloor \theta \rfloor+1}\supset \cdots \supset \Lambda_i\supset \cdots$ shrinks to the origin and hence completes the proof. 
\end{newproofof}

\begin{newproofof}
	\textit{Theorem \ref{thm_disc}.}
	In view of (\ref{convergence eq discrete}) which is valid in $[t_{i-1},t_i)$, we conclude that $\dot{V}(x)\leqslant-\bar{{c}}\bar{\sigma}(|x|)+\hat{\kappa}_\theta$ for all $t\in\mathbb{R}_{\geq 0}$. Hence we have $\dot{V}(x)<0$ for $|x|>\bar{\sigma}^{-1}({\hat{\kappa}_\theta}/\bar{{c}})$. Then we just need to show conditions (a)-(c) in Definition \ref{practical stability} hold. To satisfy condition (a) we can choose $\delta(\epsilon)$ such that $0<\delta(\epsilon)<\epsilon$. To satisfy condition (b) we can choose $\upsilon(r)=\bar{\sigma}^{-1}({\hat{\kappa}_\theta}/\bar{{c}})$ for $r\leqslant\bar{\sigma}^{-1}({\hat{\kappa}_\theta}/\bar{{c}})$
	and $\upsilon(r)>r$ for $r>\bar{\sigma}^{-1}({\hat{\kappa}_\theta}/\bar{{c}})$. For condition (c) we consider two cases. For $r\leqslant\bar{\sigma}^{-1}({\hat{\kappa}_\theta}/\bar{{c}})$ we can choose $T(r,\epsilon)$ to be any positive number since the trajectories do not leave the ball $\{x\in\mathbb{R}^n|~|x|\leqslant\bar{\sigma}^{-1}({\hat{\kappa}_\theta}/\bar{{c}})\}$ and hence $|x(t)|<\epsilon$ for all $t\geqslant0$ and all $\epsilon>\bar{\sigma}^{-1}({\hat{\kappa}_\theta}/\bar{{c}})$. For $r>\bar{\sigma}^{-1}({\hat{\kappa}_\theta}/\bar{{c}})$, choose $T^\prime$ such that $|x(T^\prime)|=\epsilon$. Then
	integrating (\ref{convergence eq discrete}) from $0$ to $T^\prime$ gives
	\setlength{\arraycolsep}{0.0em}
	\begin{eqnarray}
	V(\epsilon)-V(x_0) &{}\leqslant{}&-\bar{{c}}\int_{0}^{T^\prime}\bar{\sigma}(|x(t)|)dt + \hat{\kappa}\Big{\{}
	\int_{t_0=0}^{t_1}{\frac{\theta^1}{1!}}dt\nonumber \\  
	 +&{}\cdots{}& + \int_{t_{i-1}}^{t_{i}}{\frac{\theta^i}{i!}}dt+\cdots+\int_{t_{N^\prime-1}}^{T^\prime}{\frac{\theta^{N^\prime}}{{N^\prime}!}}dt\Big{\}}.\nonumber
	\end{eqnarray} 
	Hence we have $V(\epsilon)-V(x_0)\leqslant -\bar{{c}}T^\prime\bar{\sigma}(|\epsilon|)+ {\max_{1\leqslant i\leqslant N^\prime}}\{t_i-t_{i-1}\}\hat{\kappa}\sum_{i=1}^{N^\prime}{\frac{\theta^i}{i!}} \leqslant  -\bar{{c}}T^\prime\bar{\sigma}(|\epsilon|)+{\max_{1\leqslant i\leqslant N^\prime}}\{t_i-t_{i-1}\}\hat{\kappa} e^\theta$,
	where we assume $N^\prime$ triggering instants (including the first one at $t_0=0$) occur until $t=T^\prime$, {\it i.e.}, $t_{N^\prime-1}=\smash{\displaystyle\operatorname{max}_{t_i\leqslant T^\prime}\{t_i\}}$. Then we can find an upper bound on $T^\prime$ as the solution to inequality $\bar{{c}}T^\prime\bar{\sigma}(|\epsilon|)\leqslant V(r)-V(\epsilon)+\smash{\displaystyle\operatorname{\max}_{1\leqslant i\leqslant N^\prime}}\{t_i-t_{i-1}\}\hat{\kappa} e^\theta$ since $V(\epsilon)-V(r)\leqslant V(\epsilon)-V(x_0)$. We remark that $\smash{\displaystyle\operatorname{\max}_{1\leqslant i\leqslant N^\prime}}\{t_i-t_{i-1}\}$ is a function of $\epsilon$ since $N^\prime$ depends on $T^\prime$ which is a function of $\epsilon$. Then one can choose $T(r,\epsilon)$ so that $T(r,\epsilon)>(V(r)-V(\epsilon)+\operatorname{max}_{1\leqslant i\leqslant N^\prime}\{t_i-t_{i-1}\}\hat{\kappa} e^\theta)/(\bar{{c}}\bar{\sigma}(|\epsilon|))$.
\end{newproofof}

\bibliographystyle{IEEEtranS}
\bibliography{creport}

\end{document}